\documentclass[twocolumn,amsthm]{autart}
\usepackage{natbib}
\usepackage[T1]{fontenc}
\usepackage[latin1]{inputenc}
\usepackage[english]{babel}
\usepackage{amsmath}
\usepackage{amsfonts}
\usepackage{graphicx}
  \graphicspath{{./Pics/}}
  \DeclareGraphicsExtensions{.pdf,.png}
\usepackage{xspace}
\usepackage{url}
\usepackage{subcaption}
\usepackage{graphbox}
\usepackage{outlines}
\usepackage{newtxtext,newtxmath}
\usepackage{algpseudocode}
\usepackage[colorlinks=true,linkcolor=black,citecolor=blue,urlcolor=blue]{hyperref}

\allowdisplaybreaks

\DeclareMathOperator{\C}{\mathcal C}

\newcommand{\eps}{\varepsilon}
\newcommand{\complclass}[1]{{\sc #1}\xspace}
\newcommand{\coNP}{\complclass{coNP}}
\newcommand{\PSpace}{\complclass{PSpace}}

\newtheorem{construction}[thm]{Construction}

\edef\endfrontmatter{%
  \unexpanded\expandafter{\endfrontmatter}%
  \noexpand\endNoHyper%
}

\journal{ArXiv}

\begin{document}
\begin{frontmatter}

\title{Verifying Weak and Strong k-Step Opacity\\ in Discrete-Event Systems\thanksref{footnoteinfo}}

\thanks[footnoteinfo]{This work is an extended version of~\cite{BalunMasopust2022wodes} that was presented at WODES 2022. Supported by the Ministy of Education, Youths and Sports under the INTER-EXCELLENCE project LTAUSA19098 and by IGA PrF 2022 018 and 2023 026.
Corresponding author: T. Masopust.}

\author{Ji{\v r}{\' i}~Balun}\ead{jiri.balun01{@}upol.cz} \and
\author{Tom{\' a}{\v s}~Masopust}\ead{tomas.masopust{@}upol.cz}

\address{Faculty of Science, Palacky University Olomouc, Czechia}

\begin{abstract}
  Opacity is an important system-theoretic property expressing whether a system may reveal its secret to a passive observer (an intruder) who knows the structure of the system but has only limited observations of its behavior. Several notions of opacity have been discussed in the literature, including current-state opacity, $k$\nobreakdash-step opacity, and infinite-step opacity.
  We investigate weak and strong $k$\nobreakdash-step opacity, the notions that generalize both current-state opacity and infinite-step opacity, and ask whether the intruder is not able to decide, at any time instant, when respectively whether the system was in a secret state during the last $k$ observable steps.
  We design a new algorithm verifying weak $k$\nobreakdash-step opacity, the complexity of which is lower than the complexity of existing algorithms and does not depend on the parameter~$k$, and show how to use it to verify strong $k$-step opacity by reducing strong $k$\nobreakdash-step opacity to weak $k$\nobreakdash-step opacity. The complexity of the resulting algorithm is again better than the complexity of existing algorithms and does not depend on the parameter $k$.
\end{abstract}


\end{frontmatter}

\section{Introduction}
  Opacity is an information-flow property used to study security and privacy questions of discrete-event systems, communication protocols, and computer systems. Besides security and privacy, an interesting application of opacity is its ability to express other state-estimation properties. In particular, \cite{Lin2011} has shown how to express and verify observability, diagnosability, and detectability in terms of opacity.
  
  Opacity guarantees that a system prevents an intruder from revealing its secret. The intruder is a passive observer that knows the structure of the system but has only limited observations of system's behavior. Intuitively, the intruder estimates the behavior of the system, and the system is opaque if for every secret behavior, there is a nonsecret behavior that looks the same to the intruder. The secret is modeled as a set of secret states or as a set of secret behaviors. Modeling the secret as a set of secret states results in {\em state-based opacity}, introduced by \cite{Bryans2005,BryansKMR08} for Petri nets and transition systems, and adapted to automata by \cite{SabooriHadjicostis2007}. Modeling the secret as a set of secret behaviors results in {\em language-based opacity}, introduced by \cite{Badouel2007} and \cite{Dubreil2008}. See the overview by \cite{JacobLF16} for more details. 

  Many notions of opacity have been discussed in the literature, including initial-state opacity and current-state opacity.
  While initial-state opacity prevents the intruder from revealing, at any time instant, whether the system started in a secret state, current-state opacity prevents the intruder ``only'' from revealing whether the system is currently in a secret state.
  This, however, does not exclude the possibility that the intruder later realizes that the system was in a secret state at a former step of the computation. For instance, if the intruder estimates that the system is in one of two possible states and, in the next step, the system proceeds by an observable event enabled only in one of the states, then the intruder reveals the state in which the system was one step ago.
  
  \cite{SabooriHadjicostis2007,SabooriH12a} addressed this issue and introduce the notion of weak $k$\nobreakdash-step opacity. Weak $k$\nobreakdash-step opacity requires that the intruder is not able to ascertain the secret in the current state and $k$ subsequent observable steps. For $k=0$ and $k=\infty$, weak $k$\nobreakdash-step opacity coincides with current-state opacity and infinite-step opacity, respectively. The notion of infinite-step opacity may be confusing in the context of finite automata, because an $n$\nobreakdash-state automaton is infinite-step opaque if and only if it is $(2^n-2)$\nobreakdash-step opaque, see~\cite{Yin2017}.

  The verification of weak $k$-step opacity has been intensively studied in the literature, resulting in (at least) five different approaches: (1) the secret observer approach with complexity $O(\ell 2^{n(k+3)})$, where $n$ is the number of states and $\ell$ is the number of observable events, (2) the reverse comparison approach with complexity $O((n+m)(k+1)3^n)$, where $m \le \ell n^2$ is the number of transitions in an involved NFA, (3) the state estimator approach of \cite{SabooriH11} with complexity $O(\ell (\ell+1)^k 2^n)$, (4) the two-way observer approach of \cite{Yin2017} with complexity $O(\min\{n2^{2n},n\ell^{k} 2^{n}\})$, already corrected by \cite{Lan2020}, and (5) the projected automaton approach of \cite{BalunMasopust2021} of complexity $O((k+1)2^n (n + m\ell^2))$; see also \cite{Wintenberg2021} for more details on the state complexity and an experimental comparison.
  The reader can see that the complexity of all the algorithms depends on the parameter $k$. A partial exception is the two-way observer algorithm that does not depend on the parameter $k$ if $\ell^k > 2^n$, that is, if $k$ is larger than the number of states divided by the logarithm of the number of observable events.

  In this paper, we design a new algorithm verifying weak $k$\nobreakdash-step opacity, the complexity of which does not depend on the parameter $k$. The state complexity of our algorithm is $O(n2^n)$ and the time complexity is $O((n+m)2^n)$, where $n$ is the number of states of the input automaton and $m\le \ell n^2$ is the number of transitions of the projected input automaton. 
  Hence, our algorithm is faster than all the existing algorithms, with the exception of a very small parameter $k$; namely, if $k$ is smaller than $2\log(n)/\log(\ell)$, where $n$ is the number of states of the input automaton and $\ell$ is its number of observable events, then the algorithms based on the state estimator and on the two-way observer are, in the worst-case, faster than our algorithm.
  
  However, \cite{Falcone2014} have later noticed that even weak $k$\nobreakdash-step opacity may not be as confidential as intuitively expected. In particular, the intruder may realize that the system was in a secret state, although it cannot deduce the exact time when this has happened (see an example in Section~\ref{strongksoAlgo} or in  \cite{Falcone2014} for more details). 
  This problem motivated \cite{Falcone2014} to introduce strong $k$\nobreakdash-step opacity as the notion of $k$\nobreakdash-step opacity with a higher level of confidentiality. The idea is that whereas weak $k$\nobreakdash-step opacity prevents the intruder from revealing the exact time when the system was in a secret state during the last $k$ observable steps, strong $k$\nobreakdash-step opacity prevents the intruder from revealing that the system was in a secret state during the last $k$ observable steps.

  \cite{ma2021} pointed out without proofs that strong and weak $k$\nobreakdash-step opacity are incomparable in the sense that neither strong $k$\nobreakdash-step opacity implies weak $k$\nobreakdash-step opacity nor vice versa. In fact, under the assumption that there are no neutral states (states that are neither secret nor nonsecret), which is assumed in most of the literature, including this paper, strong $k$\nobreakdash-step opacity implies weak $k$\nobreakdash-step opacity (see Appendix~\ref{ma} for more details).

  However, the verification of one type of $k$\nobreakdash-step opacity cannot be directly used for the verification of the other. We show how to do it indirectly. Namely, we construct a polynomial-time transformation of strong $k$-step opacity to weak $k$-step opacity, which makes it possible to verify strong $k$\nobreakdash-step opacity by the algorithms for weak $k$\nobreakdash-step opacity.
  In addition, using our new algorithm verifying weak $k$\nobreakdash-step opacity results in a new algorithm to verify strong $k$\nobreakdash-step opacity with a lower complexity that does not depend on the parameter $k$.

\section{Preliminaries}
  We assume that the reader is familiar with discrete-event systems, see \cite{Lbook} for more details. For a set $S$, $|S|$ denotes the cardinality of $S$, and $2^{S}$ denotes the power set of $S$. An alphabet $\Sigma$ is a finite nonempty set of events. A string over $\Sigma$ is a sequence of events from $\Sigma$; the empty string is denoted by $\varepsilon$. The set of all finite strings over $\Sigma$ is denoted by $\Sigma^*$. A language $L$ over $\Sigma$ is a subset of $\Sigma^*$. 
  For a string $u \in \Sigma^*$, $|u|$ denotes the length of $u$.
  
  A {\em nondeterministic finite automaton\/} (NFA) over an alphabet $\Sigma$ is a structure $G = (Q,\Sigma,\delta,I,F)$, where $Q$ is a finite set of states, $I\subseteq Q$ is a nonempty set of initial states, $F \subseteq Q$ is a set of marked states, and $\delta \colon Q\times\Sigma \to 2^Q$ is a transition function that can be extended to the domain $2^Q\times\Sigma^*$ by induction; we often consider the transition function $\delta$ as the corresponding relation $\delta \subseteq Q\times \Sigma\times Q$. In addition, for a set $S\subseteq \Sigma^*$, we define $\delta(Q,S) = \cup_{s\in S}\,\delta(Q, s)$.
  For a set $Q_0\subseteq Q$, the set $L_m(G,Q_0) = \{w\in \Sigma^* \mid \delta(Q_0,w)\cap F \neq\emptyset\}$ is the language marked by $G$ from the states of $Q_0$, and $L(G,Q_0) = \{w\in \Sigma^* \mid \delta(Q_0,w)\neq\emptyset\}$ is the language generated by $G$ from the states of $Q_0$. The languages {\em marked\/} and {\em generated\/} by $G$ are defined as $L_m(G)=L_m(G,I)$ and $L(G)=L(G,I)$, respectively.
  The NFA $G$ is {\em deterministic\/} (DFA) if $|I|=1$ and $|\delta(q,a)|\le 1$ for every $q\in Q$ and $a \in \Sigma$. In this case, we identify the singleton $I=\{q_0\}$ with its element $q_0$, and simply write $G=(Q,\Sigma,\delta,q_0,F)$ instead of $G=(Q,\Sigma,\delta,\{q_0\},F)$.

  A {\em discrete-event system} (DES) $G$ over $\Sigma$ is an NFA over $\Sigma$ together with the partition of $\Sigma$ into $\Sigma_o$ and $\Sigma_{uo}$ of {\em observable\/} and {\em unobservable events}, respectively. If we want to specify that the DES is modeled by a DFA, we talk about {\em deterministic\/} DES. If the marked states are irrelevant, we omit them and simply write $G=(Q,\Sigma,\delta,I)$.

  The state estimation is modeled by {\em projection\/} $P\colon \Sigma^* \to \Sigma_o^*$, which is a morphism for concatenation defined by $P(a) = \varepsilon$ if $a\in \Sigma_{uo}$, and $P(a)= a$ if $a\in \Sigma_o$. The action of $P$ on a string $a_1a_2\cdots a_n$ is to erase all unobservable events, that is, $P(a_1a_2\cdots a_n)=P(a_1)P(a_2) \cdots P(a_n)$. The definition can be readily extended to languages.

  Let $G$ be a DES over $\Sigma$, and let $P\colon \Sigma^* \to \Sigma_o^*$ be the corresponding projection. The {\em projected automaton\/} of $G$ is the NFA $P(G)$ obtained from $G$ by replacing every transition $(p,a,q)$ by $(p,P(a),q)$, followed by the standard elimination of the $\eps$-transitions. In particular, if $\delta$ is the transition function of $G$, then the transition function $\gamma\colon Q\times \Sigma_o \to 2^Q$ of $P(G)$ is defined as $\gamma(q,a)=\delta(q,P^{-1}(a))$. The projected automaton $P(G)$ is an NFA over $\Sigma_o$ with the same states as $G$ that recognizes the language $P(L_m(G))$ and that can be constructed in polynomial time, see \cite{HopcroftU79}.

  We call the DFA constructed from $P(G)$ by the standard subset construction a {\em full observer\/} of $G$. The accessible part of the full observer of $G$ is called an {\em observer\/} of $G$, cf. \cite{Lbook}. The full observer has exponentially many states compared with $G$. In the worst case, the same holds for the observer as well, see \cite{wong98,JiraskovaM12} for more details.

  For two DESs $G_i=(Q_i,\Sigma,\delta_i,I_i)$, $i=1,2$, over the common alphabet $\Sigma$, the {\em product automaton\/} of $G_1$ and $G_2$ is defined as the DES $G_1\times G_2=(Q_1\times Q_2, \Sigma, \delta, I_1\times I_2)$, where $\delta((q_1,q_2),a)=\delta_1(q_1,a)\times \delta_2(q_2,a)$, for every pair of states $(q_1,q_2)\in Q_1\times Q_2$ and every event $a\in \Sigma$. Notice that the definition does not restrict the state space of the product automaton to its reachable part.

\section{Verification of Weak {\it k}-Step Opacity}\label{ksoAlgo}
  In this section, we recall the definition of weak $k$-step opacity for DESs, and design a new algorithm to verify weak $k$\nobreakdash-step opacity.
  To this end, we denote by $\mathbb{N}_\infty = \mathbb{N}\cup\{\infty\}$ the set of all nonnegative integers together with their limit. For $k\in\mathbb{N}_\infty$, weak $k$\nobreakdash-step opacity asks whether the intruder cannot reveal the secret of a system in the current and $k$ subsequent states.

  \begin{defn}
    Given a DES $G=(Q,\Sigma,\delta,I)$ and $k\in\mathbb{N}_\infty$. System $G$ is {\em weakly $k$\nobreakdash-step opaque ($k$-SO)\/} with respect to the sets $Q_{S}$ of secret and $Q_{NS}$ of nonsecret states and observation $P\colon \Sigma^*\to\Sigma_o^*$ if for every string $st \in L(G)$ with $|P(t)| \leq k$ and $\delta(\delta(I,s)\cap Q_S, t) \neq \emptyset$, there exists a string $s't' \in L(G)$ such that $P(s)=P(s')$, $P(t)=P(t')$, and $\delta (\delta(I,s')\cap Q_{NS}, t') \neq \emptyset$.
  \end{defn}
  
  \begin{algorithm}\hrule height .08em\vspace{-5pt}
    \caption{Verification of weak $k$-step opacity}\label{alg1}
    \begin{algorithmic}[1]
      \vspace{2pt}\hrule\vspace{3pt}
      \Require A DES $G=(Q,\Sigma,\delta,I)$, $Q_S,Q_{NS}\subseteq Q$, $\Sigma_o\subseteq \Sigma$, and $k\in\mathbb{N}_\infty$.

      \Ensure {\tt true} if and only if $G$ is weakly $k$-step opaque with respect to $Q_S$, $Q_{NS}$, and $P\colon \Sigma^*\to\Sigma_o^*$

      \State Set $Y := \emptyset$\label{l1}
      \State Compute the observer $G^{obs}$ of $G$\label{l2}
      \State Compute the projected automaton $P(G)$ of $G$\label{l3}
      \For{every state $X$ of $G^{obs}$}\label{l4}
        \For{every state $x\in X\cap Q_S$}
          \State add state $(x,X \cap Q_{NS})$ to set $Y$\label{line6}
        \EndFor
      \EndFor\label{l8}
      \State Construct $H$ as the part of the full observer of $G$ accessible from the states of the second components of $Y$\label{l9}
      
      \State Compute the product automaton $\C = P(G) \times H$\label{l10}
      
      \State Use the Breadth-First Search (BFS) of Algorithm~\ref{bfsNew} to mark all states of $\C$ reachable from the states of $Y$ in at most $k$ steps\label{l11}
      
      \If {$\C$ contains a marked state of the form $(q,\emptyset)$}
        \State \Return {\tt false}
        \Else
        \State \Return {\tt true}
      \EndIf
    \end{algorithmic}
    \hrule height .08em
  \end{algorithm}
  
  Algorithm~\ref{alg1} describes our new algorithm verifying weak $k$\nobreakdash-step opacity.
  The idea of the algorithm is as follows. We first compute the observer of $G$, denoted by $G^{obs}$, and the projected automaton of $G$, denoted by $P(G)$. Then, for every reachable state $X$ of $G^{obs}$, we add the pairs $(x,X\cap Q_{NS})$ to the set $Y$, where $x$ is a secret state of $X$ and $X\cap Q_{NS}$ is the set of all nonsecret states of $X$. 
  Intuitively, in these states, the intruder estimates that $G$ may be in the secret state $x$ or in the nonsecret states of $X\cap Q_{NS}$. To verify that the intruder does not reveal the secret state, we need to check that every possible path of length up to $k$ starting in $x$ is accompanied by a path with the same observation starting in a nonsecret state of $X\cap Q_{NS}$. 
  To this end, we construct the automaton $H$ as the part of the full observer of $G$ consisting only of states reachable from the states forming the second components of the pairs in $Y$, and the automaton $\C=P(G) \times H$ as the product automaton of the projected automaton of $G$ and $H$. In $\C$, all transitions are observable, and every path from a secret state $x$ is synchronized with all the possible paths with the same observation starting in the states of $X\cap Q_{NS}$. Thus, if there is a path from the secret state $x$ of length up to $k$ that is not accompanied by a path with the same observation from a state of $X\cap Q_{NS}$, then this path from the state $x$ in $P(G)$ ends up in a state, say, $q$, whereas all paths in $H$ with the same observation from the state $X\cap Q_{NS}$ end up in the state $\emptyset$. Here, $X\cap Q_{NS}$ and $\emptyset$ are understood as the states of the full observer of $G$. Thus, if the DES $G$ is not weakly $k$-step opaque, there is a state of $Y$ from which a state of the from $(q,\emptyset)$ is reachable in at most $k$ steps. 
  Therefore, we search the automaton $\C$ and mark all its states that are reachable from a state of $Y$ in at most $k$ steps. If a state of the from $(q,\emptyset)$ is marked, then $G$ is not weakly $k$-step opaque; otherwise, it is.
  
  We prove the correctness of Algorithm~\ref{alg1} and analyze its complexity in detail below. Intuitively, the correctness follows from the fact that the BFS visits all nodes at distance $d$ before visiting any nodes at distance $d+1$. In other words, all states of $\C$ reachable from the states of $Y$ in at most $k$ steps are visited (and marked) before any state at distance $k+1$. The implementation of the BFS is, however, the key step to obtain the claimed complexity.
  Namely, the classical BFS of \cite{IntroToAlg} maintains an array to store the shortest distances (aka the number of hops) of every node to an initial node. Since storing a number less than or equal to $k$ requires $\log(k)$ bits, using the classical BFS requires the space of size $O(\log(k) n 2^n)$ to store the shortest distance of every state of $\C$ to a state of $Y$, because $\C$ has $O(n2^n)$ states.

  For our purposes, we do not need to know the shortest distance of every state to a state of $Y$, but we rather need to keep track of the number of hops from the states of $Y$ made so far. 
  
  We can achieve this by modifying the classical BFS so that we do not store the shortest distances for every state of $\C$, but only the current distance. We store the current distance in the queue used by the BFS, see Algorithm~\ref{bfsNew}. In particular, we first push number 0 to the queue, followed by all the states of $Y$. Assuming that $k>0$, number 0 is processed in such a way that it is dequeued from the queue, and number 1 is enqueued.
  After processing all the states of $Y$ from the queue, that is, having number 1 at the head of the queue, we know that all the elements of the queue after number 1 are the states at distance one from the states of $Y$ and not less. 
  The algorithm proceeds this way until it has either visited all the states of $\C$ or the number stored in the queue is $k$.
  The algorithm marks all states of $\C$ that it visits.
  
  This approach requires to store only one $\log(k)$\nobreakdash-bit number at a time rather than $n2^n$ such numbers, and hence the complexity of the algorithm then basically follows from the fact that the distance is bounded by the number of states of $\C$, and not by the parameter $k$. 
  
  Since Algorithm~\ref{bfsNew} is a minor modification of the BFS of \cite{IntroToAlg}, very similar arguments show its correctness and complexity. For this reason, we do not further discuss the correctness and complexity of Algorithm~\ref{bfsNew}.
  
  \begin{algorithm}\hrule height .08em\vspace{-5pt}
    \caption{The Breadth-First Search used in Algorithm~\ref{alg1}}\label{bfsNew}
    \begin{algorithmic}[1]
      \vspace{2pt}\hrule\vspace{3pt}
      \Require A DES $G=(V,\Sigma,\delta,I)$, a set $S\subseteq V$, $k\in\mathbb{N}_\infty$

      \Ensure $G$ with all states at distance at most $k$ from the states of $Y$ marked

      \State Initialize the queue $Q := \emptyset$
      \State Enqueue number 0 to $Q$
      \State Mark every node $s\in S$ and enqueue it to $Q$
      \State Color every node $u\in V-S$ {\it white}
      \While {$Q \neq \emptyset$}
        \State $u := $ {\sc Dequeue}$(Q)$
        \If {$u \notin V$ and $u=k$}
          \State Terminate, states at distance $\le k$ were visited
        \ElsIf {$u\notin V$ and $u < k$}
          \State Enqueue $u+1$ to $Q$
        \ElsIf {$u\in V$ is a state of $G$}
          \For {every state $v$ reachable in one step from $u$}
            \If {the color of $v$ is {\it white}}
              \State Mark state $v$ and enqueue it to $Q$
            \EndIf
          \EndFor
          \State Color $u$ {\it black}
        \EndIf
      \EndWhile
    \end{algorithmic}
    \hrule height .08em
  \end{algorithm}
  
  Before we prove Theorem~\ref{thm-correctAlg1} below showing that $G$ is weakly $k$\nobreakdash-step opaque if and only if no state of the form $(\cdot,\emptyset)$ is marked in $\C$, we illustrate Algorithm~\ref{alg1} in the following two examples.
  
  In the first example, we consider one-step opacity of the DES $G$ depicted in Figure~\ref{fig:inso-cso-ex}(a) where all events are observable, state $2$ is secret, and state $4$ is nonsecret. The other states are neutral, meaning that they are neither secret nor nonsecret.\footnote{The meaning of neutral states is not yet clear in the literature. They are fundamental in language-based opacity, but questionable in state-based opacity. In any case, we cannot simply handle neutral states as nonsecret states.}
  The observer $G^{obs}$ of $G$ is depicted in Figure~\ref{fig:inso-cso-ex}(b).
  \begin{figure}
    \centering
    \subfloat[A DES $G$.]{\includegraphics[align=c,scale=.8]{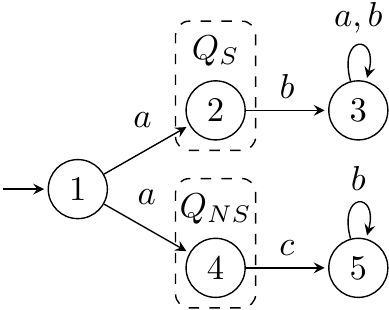}}\qquad
    \subfloat[Automata $G^{obs}$ and $H$.]{\includegraphics[align=c,scale=.8]{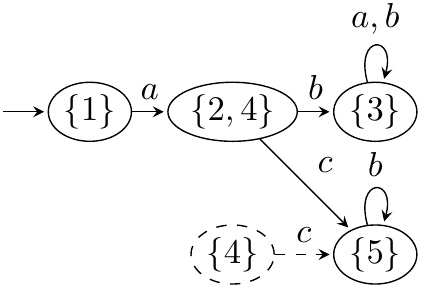}}
    \caption{A DES $G$ (a) and the observer $G^{obs}$ (b), the solid part. The automaton $H$ forming the relevant part of the full observer of $G$ is obtained from $G^{obs}$ by adding the dashed part; neither state $\emptyset$ nor the missing transitions to it are depicted in $G^{obs}$ and $H$.}
    \label{fig:inso-cso-ex}
  \end{figure}
  Since $G$ has no unobservable events, the projected automaton $P(G)=G$.
  Now, only the state $X=\{2,4\}$ of $G^{obs}$ contains a secret state, and hence intersecting it with $Q_S$ results in the set $Y=\{(2,\{4\})\}$. Notice that state $\{4\}$ is not in the observer $G^{obs}$, and therefore we need to add it to $H$ together with all the states that are reachable from state $\{4\}$ in the full observer of $G$. 
  The resulting automaton $H$ is depicted in Figure~\ref{fig:inso-cso-ex}(b) and is formed by the observer $G^{obs}$ together with the dashed state $\{4\}$ and the dashed transition from $\{4\}$ to $\{5\}$. Notice that, by the definition of the (full) observer, all the missing transitions in Figure~\ref{fig:inso-cso-ex}(b) indeed lead to state $\emptyset$, for instance, $\delta(\{1\},b)=\delta(\{5\},a)=\emptyset$. However, to keep the figures simple, we do not depict state $\emptyset$ and the transitions to state $\emptyset$.
  The marked part of the automaton $\C_1 = P(G) \times H$ reachable from the states of $Y$ in at most one step is depicted in Figure~\ref{ex:alg1-obs}. Since state $(3,\emptyset)$ is marked in $\C_1$, $G$ is not one\nobreakdash-step opaque; indeed, observing the string $ab$, the intruder reveals that $G$ must have been in the secret state $2$ one step ago.
  \begin{figure}
    \centering
    \includegraphics[scale=.8]{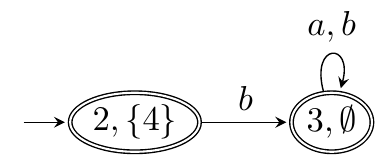}
    \caption{The reachable part of $\C_1$, where the single state of $Y$ is denoted by the little arrow.}
    \label{ex:alg1-obs}
  \end{figure}
  
  To illustrate an affirmative case, we again consider the DES $G$, but this time we assume that the event $c$ is unobservable. We denote by $\tilde{G}$ the DES $G$ where events $a,b$ are observable, the event $c$ is unobservable, state $2$ is secret, and state $4$ is nonsecret.
  The projected automaton $P(\tilde{G})$ and the observer $\tilde G^{obs}$ are depicted in Figure~\ref{ex:alg1-obs2}.
  The only state of $\tilde{G}^{obs}$ containing a secret state is the state $X=\{2,4,5\}$, which results in the set $Y=\{(2,\{4\})\}$. Again, state $\{4\}$ is not in $\tilde G^{obs}$, and hence we construct the relevant part $\tilde H$ of the full observer of $\tilde G$ by extending $\tilde G^{obs}$ by state $\{4\}$ and all the reachable states from it. The result (without state $\emptyset$ and the transitions to state $\emptyset$) is depicted in Figure~\ref{ex:alg1-obs2}(b), both the solid and the dashed part.
  The marked part of $\C_2 = P(\tilde{G}) \times \tilde{H}$ is depicted in Figure~\ref{ex:alg1-obs3}. Since no state of the form $(\cdot,\emptyset)$ is marked in $\C_2$, $\tilde{G}$ is one-step opaque.

  \begin{figure}
    \centering
    \subfloat[The automaton $P(\tilde G)$.]{\includegraphics[scale=.76]{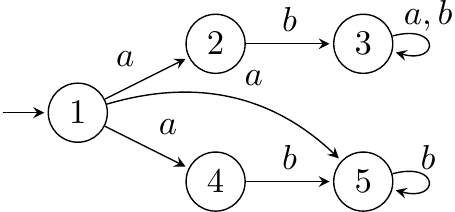}}\quad
    \subfloat[Automata $\tilde G^{obs}$ and $\tilde H$]{\includegraphics[scale=.77]{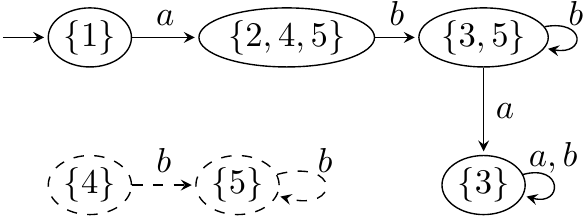}}
    \caption{The automaton $P(\tilde{G})$ (a) and the observer $\tilde G^{obs}$ (b), the solid part. The automaton $\tilde H$ forming the relevant part of the full observer of $\tilde G$ is obtained from $\tilde G^{obs}$ by adding the dashed part; neither state $\emptyset$ nor the missing transitions to it are depicted in $G^{obs}$ and $H$.}
    \label{ex:alg1-obs2}
  \end{figure}

  \begin{figure}
    \centering
    \includegraphics[scale=.8]{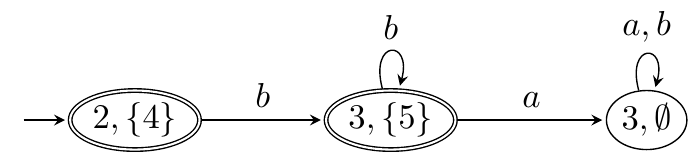}
    \caption{The reachable part of $\C_2$, where the single state of $Y$ is denoted by the little arrow.}
    \label{ex:alg1-obs3}
  \end{figure}

  We now prove the correctness of our algorithm.
  \begin{thm}\label{thm-correctAlg1}
    A DES $G$ is weakly $k$-step opaque with respect to $Q_S$, $Q_{NS}$, and $P$ if and only if Algorithm~\ref{alg1} returns {\tt true}.
  \end{thm}
  \begin{proof}
    If $G=(Q,\Sigma,\delta,I)$ is not weakly $k$\nobreakdash-step opaque, then there is $st\in L(G)$ such that $|P(t)|\le k$, $\delta(\delta(I,s)\cap Q_S, t) \neq \emptyset$, and $\delta(\delta(I,P^{-1}P(s))\cap Q_{NS},P^{-1}P(t)) = \emptyset$. We have two cases.

    (i) If $\delta(I,P^{-1}P(s))\cap Q_{NS} = \emptyset$, then $G$ is not weakly $k$\nobreakdash-step opaque. Algorithm~\ref{alg1} detects this case, because for the state $X=\delta(I,P^{-1}P(s))$ of the observer of $G$, we have that $X\cap Q_{S} \supseteq \delta(I,s)\cap Q_S \neq \emptyset$ and $X\cap Q_{NS} = \emptyset$, and hence there is $q\in X\cap Q_{S}$ resulting in adding the pair $(q,\emptyset)$ to the set $Y$ in line~\ref{line6}.
 
    (ii) If $\delta(I,P^{-1}P(s))\cap Q_{NS} = Z \neq \emptyset$, then all pairs of the form $(\delta(I,P^{-1}P(s))\cap Q_S)\times \{Z\}$ are added to $Y$. Since $\delta(\delta(I,s)\cap Q_S, t) \neq \emptyset$, there is a pair $(z,Z)\in  Y$ such that generating the string $P(t)$ in the automaton $P(G)$ from state $z$ changes the state to a state $q$. On the other hand, $\delta(Z,P^{-1}P(t))=\emptyset$ implies that generating $P(t)$ in the full observer of $G$ from state $Z$ changes the state to state $\emptyset$, and hence the pair $(q,\emptyset)$ is reachable in $\C$ from the state $(z,Z)\in Y$ in at most $|P(t)|\le k$ steps.
    In both cases, Algorithm~\ref{alg1} marks $(q,\emptyset)$, and returns {\tt false}.

    On the other hand, if $G$ is weakly $k$-step opaque, we show that no pair of the form $(q,\emptyset)$ is reachable in $\C$ from a state of $Y$ in at most $k$ steps. For the sake of contradiction, we assume that a pair $(q,\emptyset)$ is marked in $\C$. However, this means that, in $G$, there is a string $s$ and a state $z\in Q$ such that $z\in \delta(I,s)\cap Q_S$, the state of the observer of $G$ reached under the string $P(s)$ is $X=\delta(I,P^{-1}P(s))$, and, for $Z = X \cap Q_{NS}$, the pair $(q,\emptyset)$ is reachable in $\C$ from the pair $(z,Z)\in Y$ by a string $w\in \Sigma_o^*$ of length at most $k$. 
    In particular, there is a string $t\in P^{-1}(w)$ such that when $G$ generates $t$, it changes its state from $z$ to $q$. Therefore, $q\in \delta(\delta(I,s)\cap Q_S,t)\neq\emptyset$. However, $\delta(\delta(I,P^{-1}P(s))\cap Q_{NS},P^{-1}(w)) = \delta(Z,P^{-1}(w)) = \emptyset$, because generating $w$ in $\C$ changes the pair $(z,Z)$ to $(q,\emptyset)$, and hence the full observer of $G$ changes its state from $Z$ to $\emptyset$ when generating $w$. This shows that $G$ is not weakly $k$-step opaque, which is a contradiction.
  \end{proof}

  We now discuss the complexity of our algorithm.
  \begin{thm}
    The space and time complexity of Algorithm~\ref{alg1} is $O(n2^n)$ and $O((n + m)2^n)$, respectively, where $n$ is the number of states of the input DES $G$ and $m$ is the number of transitions of $P(G)$. In particular, $m \le \ell n^2$, where $\ell$ is the number of observable events.
  \end{thm}
  \begin{proof}
    Computing the observer and the projected NFA of $G$, lines~\ref{l2} and~\ref{l3}, takes time $O(\ell 2^n)$ and $O(m+n)$, respectively. The cycle on lines~\ref{l4}--\ref{l8} takes time $O(n 2^n)$. Constructing the part $H$ of the full observer of $G$, line~\ref{l9}, takes time $O(\ell 2^n)$. Constructing $\C$, line~\ref{l10}, takes time $O(n 2^n + m 2^n)$, where $O(n 2^n)$ is the number of states and $O(m 2^n)$ is the number of transitions of $\C$. The bounds come from the fact that we create at most $2^n$ copies of the automaton $P(G)$. The BFS takes time linear in the size of $\C$, and the condition of line~\ref{l11} can be processed during the BFS. Since $m\ge \ell$, the proof is complete.
  \end{proof}

  We now briefly review the complexity of existing algorithms verifying weak $k$-step opacity. First, notice that the complexity of existing algorithms is exponential, which seems unavoidable because the problem is \PSpace-complete, see \cite{BalunMasopust2021,BalunMasopust2022smc} for more details.\footnote{It is a long-standing open problem of computer science whether \PSpace-complete problems can be solved in polynomial time.}
  In particular, \cite{SabooriH11} designed an algorithm with complexity $O(\ell (\ell+1)^k 2^n)$, where $n$ is the number of states and $\ell$ is the number of observable events. Considering the verification of weak $\infty$\nobreakdash-step opacity, \cite{SabooriH12a} designed an algorithm with complexity $O(\ell 2^{n^2+n})$.
  \cite{Yin2017} introduced the notion of a two\nobreakdash-way observer and applied it to the verification of weak $k$\nobreakdash-step opacity with complexity $O(\min\{n2^{2n},n\ell^{k} 2^{n}\})$, and to the verification of weak $\infty$\nobreakdash-step opacity with complexity $O(n2^{2n})$; the formulae already include a correction by \cite{Lan2020}.
  \cite{BalunMasopust2021} designed algorithms verifying weak $k$\nobreakdash-step opacity and weak $\infty$-step opacity with complexities $O((k+1)2^n (n + m\ell^2))$ and $O((n + m\ell)2^n)$, respectively, where $m \le \ell n^2$ is the number of transitions in the projected automaton. These algorithms outperform the two\nobreakdash-way observer if $k$ is polynomial in $n$ or larger than \mbox{$2^n-2$}, since weak $(2^n-2)$-step opacity and weak $\infty$-step opacity coincide, see \cite{Yin2017}.
  \cite{Wintenberg2021} discussed and experimentally compared four approaches to the verification of weak $k$\nobreakdash-step opacity based on (i) the secret observer, (ii) the reverse comparison, (iii) the state estimator, and (iv) the two-way observer. Their respective state complexities are $O(2^{n(k+3)})$, $O(n(k+1)3^n)$, $O((\ell+1)^k 2^n)$, and $O(\min\{2^n,\ell^k\}2^n)$.\footnote{The state complexity of the two-way observer is correct. The correction of \cite{Lan2020} consists in adding a time bound to compute the intersection of two sets, and hence it does not influence the number of states.}

  Notice that these bounds are formulated only in the number of states of the constructed automata, disregarding the number of transitions and the time of the construction. Therefore, the time-complexity bounds differ from the state-complexity bounds at least by the factor of $\ell$, if the constructed automata are deterministic, or by a factor of $m\le \ell n^2$ if the construction of the automaton involves an NFA, such as in the case of the reverse comparison. 
  Namely, the time-complexity bounds are $O(\ell 2^{n(k+3)})$ for the secret observer, where $n$ is the number of states and $\ell$ is the number of observable events, $O((n+m)(k+1)3^n)$ for the reverse comparison, where $m \le \ell n^2$ is the number of transitions in an involved NFA, $O(\ell (\ell+1)^k 2^n)$ for the state estimator, and $O(\min\{n2^{2n},n\ell^{k} 2^{n}\})$ for the two-way observer.

  As the reader may notice, the above complexities depend on the parameter $k$. A partial exception is the two\nobreakdash-way observer that does not depend on $k$ if $\ell^k \ge 2^n$, that is, if $k$ is larger than the number of states divided by the logarithm of the number of observable events.

  Since the complexity of Algorithm~\ref{alg1} is $O((n+m)2^n)$, where $n$ is the number of states of the input DES $G$ and $m\le \ell n^2$ is the number of transitions of the projected automaton of $G$, it does not depend on the parameter $k$ and, in general, outperforms the existing algorithms. An exception is the case of a very small parameter $k$. In particular, if $k < 2\log(n)/\log(\ell)$, the algorithms based on the state estimator and on the two-way observer are, in the worst-case, faster than our algorithm. Notice that this theoretical result agrees with the experimental results of \cite{Wintenberg2021}.

\section{Verification of Strong {\it k}-Step Opacity}\label{strongksoAlgo}
  Although weak $k$-step opacity seems confidential enough, \cite{Falcone2014} have pointed out that it is actually not as confidential as intuitively expected. Namely, the intruder may realize that the system previously was in a secret state, though it is not able to deduce the exact time when that happened, see \cite{Falcone2014} for more details and examples. 
  Consequently, they defined strong $k$\nobreakdash-step opacity as a variation of $k$\nobreakdash-step opacity with a higher level of confidentiality.

  Before we recall the definition of strong $k$-step opacity as formulated by \cite{Falcone2014}, we illustrate the problem with weak $k$-step opacity. To this end, we consider the system depicted in Figure~\ref{exampleFig}, where state $2$ is secret and the other states are nonsecret, and where the event $a$ is observable and the event $u$ is unobservable. Observing the string $aa$, the intruder realizes that the system must have been in the secret state $2$, though it cannot say whether it was one or two steps ago. Actually, even observing the string $a$ already reveals that the system currently is or one step ago was in the secret state $2$.

  \begin{figure}
    \centering
    \includegraphics[scale=.8]{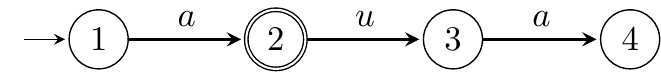}
    \caption{A deterministic DES that is not confidential enough. The secret state is double circled.}
    \label{exampleFig}
  \end{figure}

  In accordance with \cite{Falcone2014}, we consider strong $k$-step opacity only for deterministic DES where all states that are not secret are nonsecret, that is, $Q_{NS}=Q-Q_{S}$. It means that every state has its own secret/nonsecret status and there are no neutral states.

  \begin{defn}
    Given a deterministic DES $G=(Q,\Sigma,\delta,q_0)$ and $k\in \mathbb{N}_\infty$. System $G$ is {\em strongly $k$\nobreakdash-step opaque ($k$-SSO)} with respect to the set $Q_S$ of secret states and observation $P \colon \Sigma^* \to \Sigma_o^*$ if for every string $s \in L(G)$, there exists a string $w \in L(G)$ such that $P(s)=P(w)$ and for every prefix $w'$ of $w$, if $|P(w)|-|P(w')| \le k$, then $ \delta(q_0,w')\notin Q_S$.
  \end{defn}
  
  Notice that whereas weak $k$\nobreakdash-step opacity prevents the intruder from revealing the exact time when the system was in a secret state during the last $k$ observable steps, strong $k$\nobreakdash-step opacity prevents the intruder from revealing that the system was in a secret state during the last $k$ observable steps.

  For an illustration, we again consider the system depicted in Figure~\ref{exampleFig}, where state $2$ is secret and event $u$ is unobservable. The system is weakly one\nobreakdash-step opaque, but not strongly one\nobreakdash-step opaque, because for $s=aua$, the only $w$ with the same observation as $s$ is $w=aua$, and hence the prefixes $w'$ for which $|P(w)|-|P(w')|\le 1$ are the strings $w'=a$, $w'=au$, and $w'=aua$. However, for $w'=a$, we obtain that $\delta(1,a)=2 \in Q_{S}$, which violates the definition of strong one\nobreakdash-step opacity.

  In fact, the system is neither strongly 0-step opaque, because for $s=au$, the only $w$ with the same observation as $s$ are the strings $au$ and $a$, and therefore the prefixes $w'$ for which $|P(w)|-|P(w')|\le 0$ is either $w'=a$ or $w'=au$. However, for $w'=a$, we obtain that $\delta(1,a)=2 \in Q_{S}$, which violates the definition of strong $0$\nobreakdash-step opacity. On the other hand, the system is obviously current\nobreakdash-state opaque. Consequently, the notions of strong 0\nobreakdash-step opacity and current-state opacity\footnote{Current\nobreakdash-state opacity is a synonym for weak 0\nobreakdash-step opacity.} do not coincide. 
  
  We show in Theorem~\ref{thm:0-sso_0-so} below that unobservable transitions from secret states to nonsecret states, as in our example, are the only issues making the difference between strong 0\nobreakdash-step opacity and weak 0-step (current-state) opacity. To this end, we define the notion of a normal DES.

\subsection{Normalization}
  In what follows, we call the systems where there are no unobservable transitions from secret states to nonsecret states {\em normal}.
  For systems that are not normal, we provide a construction to normalize them, that is, we eliminate unobservable transitions from secret states to nonsecret states without affecting the property of being strongly $k$-step opaque.

 \begin{construction}\label{consNorm}
    Let $G=(Q,\Sigma, \delta, q_0)$ be a deterministic DES, $k\in \mathbb{N}_\infty$, $Q_S\subseteq Q$ be the set of secret states, and $P\colon \Sigma^* \to \Sigma_o^*$ be the observations. We construct
    \[
      G_{norm}=(Q_n,\Sigma, \delta_{n}, q_0)
    \]
    where $Q_n = Q\cup Q'$ for $Q'=\{q' \mid q\in Q\}$ being a disjoint copy of $Q$, and the transition function $\delta_n$ is defined as follows.
    We initialize $\delta_{n} := \delta$ and further modify it in the following four steps:
    \begin{enumerate}
      \item For every $p\in Q_S$, $q\in Q_{NS}$, and $u\in\Sigma_{uo}$, we replace the transition $(p,u,q)$ by $(p,u,q')$ in $\delta_{n}$.

      \item For every unobservable transition $(p,u,q)$ in $\delta$, that is, $u\in \Sigma_{uo}$, we add the transition $(p',u,q')$ to $\delta_n$.

      \item For every observable transition $(q,a,r)$ in $\delta$, that is, $a\in\Sigma_o$, we add the transition $(q',a,r)$ to $\delta_{n}$.
      
      \item We remove unreachable states and corresponding transitions.
    \end{enumerate}
    The set of secret states of $G_{norm}$ is the set $Q_n^S=Q_S\cup Q'$.
    \hfill$\diamond$
  \end{construction}
  
  In the sequel, we call $G_{norm}$ the {\em normalization\/} of $G$. If $G$ and $G_{norm}$ coincide, we say that $G$ is {\em normal}.
  
  To illustrate Construction~\ref{consNorm}, consider the system depicted in Figure~\ref{fig:lemma} (left). Its normalization $G_{norm}$ is depicted in the same figure (right). States $2$ and $3$ of $G$ are secret, events $a$ and $b$ are observable, and $u$ is unobservable. The normalization $G_{norm}$ of $G$ initially contains five new secret states $1'$, $2'$, $3'$, $4'$, $5'$. Step (1) of Construction~\ref{consNorm} replaces transitions $(2,u,4)$ and $(3,u,4)$ by $(2,u,4')$ and $(3,u,4')$, respectively, step (2) adds four unobservable transitions $(1',u,2')$, $(2',u,4')$, $(3',u,4')$, and $(4',u,5')$, and step (3) adds the observable transitions $(1',a,3)$, $(2',a,4)$, $(4',a,5)$ and $(5',b,5)$. Finally, step (4) eliminates unreachable states $1'$, $2'$, $3'$, and the corresponding transitions.

  \begin{figure}
    \centering
    \includegraphics[align=c,scale=.65]{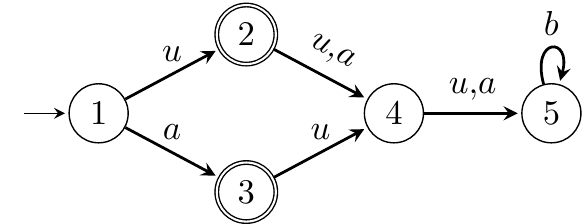}
    $\Rightarrow$
    \includegraphics[align=c,scale=.65]{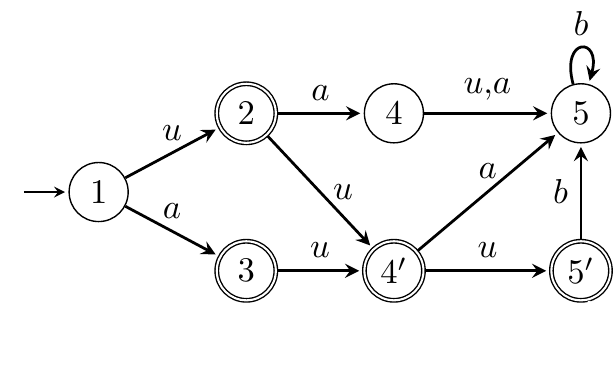}
    \caption{A deterministic DES $G$ (left) and its normalization $G_{norm}$ (right). Secret states are double circled.}
    \label{fig:lemma}
  \end{figure}

  The following lemma compares the behaviors of $G$ and its normalization $G_{norm}$.
  \begin{lem}\label{lemma:delta_norm}
    Let $G=(Q,\Sigma, \delta, q_0)$ be a deterministic DES, and let $Q_S\subseteq Q$ be the set of secret states. Let $G_{norm}$ be the normalization of $G$ obtained by Construction~\ref{consNorm}. Then, for every $w\in\Sigma^*$ and $a\in\Sigma$, the following holds:
    \begin{enumerate}
      \item For $a\in\Sigma_{uo}$, $\delta(q_0,wa)=p$ if and only if $\delta_n(q_0,wa)\in\{p,p'\}$, where $p'\in Q'$ is a copy of $p\in Q$;

      \item\label{it2} For $a\in\Sigma_o$, $\delta(q_0,wa)=\delta_n(q_0,wa)$; 

      \item $L(G)=L(G_{norm})$.
    \end{enumerate}
  \end{lem}
  \begin{proof}
    We prove (1) and (2) by induction on the length of $w$. The induction hypothesis is that either $\delta(q_0,w)=p=\delta_n(q_0,w)$, or $\delta(q_0,w)=p$ and $\delta_n(q_0,w)=p'$.

    To prove (1), let $a$ be unobservable. We first consider the case $\delta(q_0,w)=\delta_n(q_0,w)=p$.
    First, if $p$ is nonsecret, Construction~\ref{consNorm} adds every transition $(p,a,q)\in \delta$ to $\delta_n$.
    On the other hand, if $p$ is secret, $\delta_n$ contains the transition $(p,a,q')$ for every transition $(p,a,q)\in \delta$ with $q\in Q_{NS}$, and the transition $(p,a,q)$ for every transition $(p,a,q)\in \delta$ with $q \in Q_S$.
    In both cases, Construction~\ref{consNorm} adds no other transition from state $p$ to $\delta_n$, and hence $\delta(q_0,wa) = \delta(p,a) = q$ if and only if $\delta_n(q_0,wa) = \delta_n(p,a) \in \{q,q'\}$. 
    Notice that this case also covers the base case of the induction, since for $w=\eps$, $\delta(q_0,w) = \delta_n(q_0,w) = q_0$.

    Now, we consider the case $\delta_n(q_0,w)=p'$ and $\delta(q_0,w)=p$. Since Construction~\ref{consNorm} adds the transition $(p',a,q')$ to $\delta_n$ for every unobservable transition $(p,a,q)\in\delta$, we have that $\delta(q_0,wa)=\delta(p,a)=q$ if and only if $\delta_n(q_0,wa)=\delta_n(p',a)=q'$.

    To prove (2), let $a$ be observable. We first consider the case $\delta(q_0,w)=\delta_n(q_0,w)=p$. Then, from the state $p$, Construction~\ref{consNorm} adds to $\delta_n$ all and only the observable transitions of $\delta$, and hence $\delta(p,a) = \delta_{n}(p,a)$.

    Now, we consider the case $\delta_n(q_0,w)=p'$ and $\delta(q_0,w)=p$. Then, Construction~\ref{consNorm} adds the transition $(p',a,q)$ to $\delta_n$ for every observable transition $(p,a,q)\in\delta$, and therefore $\delta(q_0,wa) = \delta(p,a) = \delta_n(p',a) = \delta_n(q_0,wa)$.

    Finally, $L(G)=L(G_{norm})$ of (3) follows from (1) and (2), since, for every $w\in\Sigma^*$, $\delta(q_0,w)$ is undefined if and only if $\delta_n(q_0,w)$ is undefined.
  \end{proof}

  The following lemma describes the meaning of normalization and states the main properties of a normalized DES.

  \begin{lem}\label{lemma:normalization}
    For a deterministic DES $G=(Q,\Sigma, \delta, q_0)$, $k\in \mathbb{N}_\infty$, the set of secret states $Q_S$, and the observation $P\colon \Sigma^* \to \Sigma_o^*$, let $G_{norm}$ be the normalization of $G$ obtained by Construction~\ref{consNorm}. Then, the following holds true:
    \begin{enumerate}
      \item $G_{norm}$ is deterministic;

      \item In $G_{norm}$, there is no nonsecret state reachable from a secret state by a sequence of unobservable events, that is, $\delta_n(Q_{n}^{S},P^{-1}(\eps))\cap (Q_n-Q_n^S) = \emptyset$;

      \item $G$ is $k$-SSO with respect to $Q_S$ and $P$ if and only if $G_{norm}$ is $k$-SSO with respect to $Q_n^S$ and $P$.
    \end{enumerate}
  \end{lem}
  \begin{proof}
    To prove $(1)$, we analyze the steps of Construction~\ref{consNorm} creating $\delta_n$. First, $\delta_n$ is defined as $\delta$, which is deterministic. Then, step (1) replaces some unobservable transitions, which is an operation that preserves determinism of $\delta_n$. Step (2) adds the transition $(p',u,q')$ for every unobservable transition $(p,u,q)$ in $G$. Similarly, step (3) adds the transition $(q',a,p)$ for every observable transition $(q,a,p)$ in $G$. Since $G$ is deterministic, steps (2) and (3) preserve determinism. Altogether, $G_{norm}$ is deterministic.

    To prove $(2)$, step (1) of Construction~\ref{consNorm} replaces all unobservable transitions from a secret state to a nonsecret state by transitions from a secret state to a new secret state. Step (2) adds unobservable transitions only between the new states, which are all secret. Since no unobservable transition is defined from the new states to the old states, there is no nonsecret state in $G_{norm}$ reachable from a secret state by a sequence of unobservable events.

    To prove the first direction of $(3)$, we assume that $G$ is $k$\nobreakdash-SSO with respect to $Q_S$ and $P$, and show that then $G_{norm}$ is $k$\nobreakdash-SSO with respect to $Q_n^S$ and $P$. To this end, we show that for every string $s\in L(G_{norm})$, there exists a string $w \in L(G_{norm})$ such that $P(s)=P(w)$ and, for every prefix $w'$ of $w$, if $|P(w)|-|P(w')|\leq k$, then $\delta_n(q_0, w')\not\in Q_n^S$.
    Thus, let $s\in L(G_{norm})$ be an arbitrary string. Then, by Lemma~\ref{lemma:delta_norm}, $s\in L(G_{norm})=L(G)$, and since $G$ is $k$\nobreakdash-SSO with respect to $Q_S$ and $P$, there is a string $\tilde w\in L(G)$ such that $P(s)=P(\tilde w)$ and, for every prefix $\tilde w'$ of $\tilde w$, if $|P(\tilde w)|-|P(\tilde w')|\leq k$, then $\delta(q_0, \tilde w')\not\in Q_S$.
    By defining $w=\tilde w$, we obtain that the string $w \in L(G)=L(G_{norm})$ and that $P(s)=P(w)$. It remains to show that for every prefix $w'$ of $w$, if $|P(w)|-|P(w')|\leq k$, then $\delta_n(q_0, w')\not\in Q_n^S$.
    To this end, let $xy=w$ be the decomposition of $w$, where $x$ is the shortest prefix of $w$ such that $|P(w)|-|P(x)|\leq k$. Then, $x$ is either empty or ends with an observable event. Hence, by Lemma~\ref{lemma:delta_norm}, $\delta(q_0,x) = \delta_n(q_0,x) = q\in Q$ in $G$.
    However, for every prefix $y'$ of $y$, the string $xy'$ is a prefix of $\tilde w$ satisfying $|P(\tilde w)|-|P(xy')|\le k$, and hence $\delta(q_0,xy')\notin Q_S$. In other words, the computation of $\delta(q,y)$ in $G$ does not go through a secret state, and therefore the same sequence of transitions exists in $G_{norm}$, that is, $\delta(q_0,xy') = \delta_n(q_0,xy') \notin Q_n^S=Q_S\cup Q'$.
    Since every prefix $w'$ of $w$ satisfying $|P(w)|-|P(w')|\leq k$ is of the form $w'=xy'$, where $y'$ is a prefix of $y$, we have shown that $\delta_n(q_0,w')\notin Q_n^S$, which was to be shown.

    To prove the other direction, we assume that $G$ is not $k$\nobreakdash-SSO with respect to $Q_S$ and $P$, and show that neither the $G_{norm}$ is $k$\nobreakdash-SSO with respect to $Q_n^S$ and $P$.
    To this end, let $s\in L(G)$ be a string violating $k$\nobreakdash-SSO in $G$; that is, for every $w\in L(G)$ such that $P(s)=P(w)$, there exists a prefix $w'$ of $w$ such that $|P(w)|-|P(w')|\leq k$ and $\delta(q_0, w')=q_w\in Q_S$.
    However, by Lemma~\ref{lemma:delta_norm}, $w\in L(G_{norm})=L(G)$ and $\delta_n(q_0,w')\in\{q_w,q_w'\}$. Since both states $q_w,q_w'\in Q_n^S=Q_S\cup Q'$ are secret in $G_{norm}$, we conclude that $G_{norm}$ is not $k$\nobreakdash-SSO with respect to $Q_n^S$ and $P$.
  \end{proof}

\subsection{Weak versus Strong 0-Step Opacity}
  In this section, we discuss the relationship between strong $0$\nobreakdash-step opacity and weak $0$-step (current-state) opacity for normal deterministic DES. The following result characterizes the relationship between these two notions and fixes the claim of \cite{ma2021} stating that strong 0\nobreakdash-step opacity reduces to current\nobreakdash-state opacity, which is not the case as shown in the example of Figure~\ref{exampleFig}.
  
  \begin{thm}\label{thm:0-sso_0-so}
    A normal deterministic DES $G=(Q,\Sigma,\delta,q_0)$ is strongly $0$\nobreakdash-step opaque with respect to $Q_S$ and $P$ if and only if $G$ is weakly $0$-step opaque with respect to $Q_S$, $Q_{NS}=Q-Q_{S}$, and $P$.
  \end{thm}
  \begin{proof}
    We first assume that $G=(Q,\Sigma,\delta,q_0)$ is $0$\nobreakdash-SSO with respect to $Q_S$ and $P$. To show that $G$ is $0$\nobreakdash-SO with respect to $Q_S$ and $P$, let $st\in L(G)$ be such that $|P(t)|\le 0$ and $\delta(q_0, s)\in Q_S$. Since $st\in L(G)$ and $G$ is deterministic, $\delta(q_0,st)$ is defined. Therefore, we need to show that there is a string $s't' \in L(G)$ such that $P(s)=P(s')$, $P(t)=P(t')$, and $\delta(q_0,s')\in Q_{NS}=Q-Q_{S}$. However, since $G$ is $k$\nobreakdash-SSO with respect to $Q_S$ and $P$, there is a string $w\in L(G)$ such that $P(w)=P(st)$ and, for every prefix $w'$ of $w$ with $|P(w)|-|P(w')|= 0$, $\delta(q_0,w')\in Q-Q_{S}$. 
    Let $w'$ be any, but fixed, such prefix of $w$. We set $s'=w'$ and $s't'=w$. Then, $P(s') = P(w') = P(w) = P(st) = P(s)$, $P(t')=\eps=P(t)$, and $\delta(q_0,s')=\delta(q_0,w')\in Q-Q_{S}=Q_{NS}$. Thus, we have shown that $G$ is $0$\nobreakdash-SO with respect to $Q_S$ and $P$.\footnote{Notice that the proof does not hold if we admit neutral states, that is, $Q_{NS}\neq Q-Q_{S}$. However, it is not the case in this paper.}
 
    For the other direction, we assume that $G$ is not $0$\nobreakdash-SSO with respect to $Q_S$ and $P$. To show that $G$ is neither $0$\nobreakdash-SO with respect to $Q_S$ and $P$, we need to find a string $st\in L(G)$ with $|P(t)|\le 0$ and $\delta(q_0, s)\in Q_S$ such that, for every string $s't' \in L(G)$ with $P(s)=P(s')$ and $P(t)=P(t')$, the state $\delta(q_0,s')\in Q_{S}$.
    However, from the assumption that $G$ is not $0$\nobreakdash-SSO with respect to $Q_S$ and $P$, we have a string $st\in L(G)$ such that $\delta(q_0, s)\in Q_S$, $|P(t)|\le 0$, and, for every string $w\in L(G)$ with $P(st)=P(w)$, there is a prefix $w'$ of $w$ such that $|P(w)|-|P(w')|=0$ and $\delta(q_0, w')\in Q_S$.  
    To complete the proof, we show that for every $s't' \in L(G)$ with $P(s)=P(s')$ and $P(t)=P(t')$, the state $\delta(q_0,s')$ is secret. 
    To this end, let $xy=s't'$ be the decomposition of $s't'$ such that $y$ is the longest suffix of $s't'$ consisting only of unobservable events. Notice that $x$ is a prefix of $s'$, because $P(t')=P(t)=\eps$. Since $P(st)=P(x)$, there must be a prefix $x'$ of $x$ such that $|P(x)|-|P(x')|=0$ and $\delta(q_0,x') \in Q_S$. However, the last event of $x$ is observable, and hence the only prefix $x'$ of $x$ for which $|P(x)|-|P(x')|=0$ is $x'=x$, and therefore $\delta(q_0,x) = \delta(q_0,x') \in Q_S$. Since $G$ is normal, there are no nonsecret states reachable from the secret state $\delta(q_0,x)$ under a sequence of unobservable events. In particular, $\delta(q_0,s')=\delta(q_0,xy')\in Q_S$, where $y'$ is a prefix of $y$ for which $s'=xy'$; recall that $y$ is the longest suffix of $s't'$ consisting only of unobservable events. 
    Thus, we have shown that $G$ is not $0$-SO.
  \end{proof}

\subsection{Weak versus Strong {\it k}-Step Opacity}
  In this section, we show how to reduce strong $k$-step opacity to weak $k$\nobreakdash-step opacity. In the construction, we assume that $G$ is a normal deterministic DES. By Lemma~\ref{lemma:normalization}, this assumption is without loss of generality, because if $G$ is not normal, then we can consider $G_{norm}$ instead.

  \begin{construction}\label{consSSOtoSO}
    Let $G=(Q, \Sigma, \delta, q_0)$ be a normal deterministic DES, $P\colon\Sigma^*\to \Sigma_o^*$ be the observation projection, and $Q_{S}$ be the set of secret states. We construct
    \[
      G'=(Q\cup Q_{NS}',\Sigma\cup\{u\},\delta', q_0)
    \]
    as a disjoint union of $G$ and $G_{ns}=(Q_{NS}',\Sigma,\delta_{ns},q_0')$, where $G_{ns}$ is obtained from $G$ by removing all secret states and corresponding transitions, and $Q_{NS}'=\{q' \mid q\in Q_{NS}\}$ is a copy of $Q_{NS}$ disjoint from $Q$.
    We use a new unobservable event $u$ to connect $G_{ns}$ to $G$ so that we 
      initialize $\delta':=\delta\cup\delta_{ns}$ and
      extend $\delta'$ by additional transitions $(q,u,q')$ for every $q\in Q_{NS}$,
     cf. Figure~\ref{fig:ksso-kso} for an illustration.
    The states of $Q_{NS}'$ are the only nonsecret states of $G'$, that is, the set of secret states of $G'$ is the set $Q_S'=Q$. Finally, we define the projection $P'\colon (\Sigma\cup\{u\})^* \to \Sigma_o^*$.
    \hfill$\diamond$
  \end{construction}

  \begin{figure}
    \centering
    \includegraphics[align=c,scale=1]{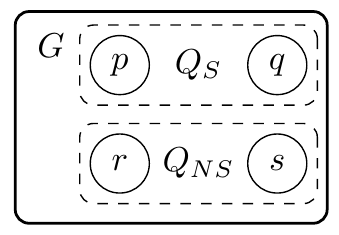}
    $\Longrightarrow$
    \includegraphics[align=c,scale=1]{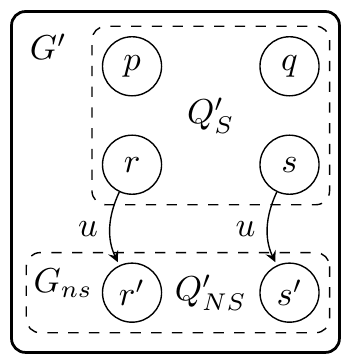}
    \caption{An illustration of Construction~\ref{consSSOtoSO} transforming strong $k$\nobreakdash-step opacity to weak $k$\nobreakdash-step opacity.}
    \label{fig:ksso-kso}
  \end{figure}

  The following theorem describes the relationship between strong $k$-step opacity and weak $k$-step opacity, and justifies the correctness of Algorithm~\ref{alg2} below.
  
  \begin{thm}\label{thm_ksso-kso}
    Let $G=(Q,\Sigma,\delta,q_0)$ be a normal deterministic DES, and let $G'$ be the DES obtained from $G$ by Construction~\ref{consSSOtoSO}. Then, $G$ is strongly $k$\nobreakdash-step opaque with respect to $Q_S$ and $P$ if and only if $G'$ is weakly $k$\nobreakdash-step opaque with respect to $Q_S'$, $Q_{NS}'$, and $P'$, where $Q_S'$, $Q_{NS}'$, and $P'$ are defined in Construction~\ref{consSSOtoSO}.
  \end{thm}
  \begin{proof}
    For the first implication, we assume that $G$ is $k$\nobreakdash-SSO with respect to $Q_S$ and $P$, and we show that $G'$ is $k$-SO with respect to $Q_S'$, $Q_{NS}'$, and $P'$. To this end, let $st \in L(G')$ be such that $|P'(t)| \leq k$ and $\delta'(q_0,s)\in Q_S'$. We need to show that there is a string $s't' \in L(G')$ such that $P'(s)=P'(s')$, $P'(t)=P'(t')$, and $\delta'(q_0,s')\in Q_{NS}'$.

    Let $P_u$ denote the projection that removes every occurrence of event $u$, that is, $P_u(a)=a$ for $a\in \Sigma$, and $P_u(u)=\eps$. 
    We first show that $P_u(st)\in L(G)$. Indeed, if $st$ does not contain $u$, then $P_u(st)=st\in L(G)$. 
    If $st$ contains $u$, then, by the construction of $G'$, any string of $L(G')$ contains at most one occurrence of $u$. Since $\delta'(q_0,s)\in Q_S'$, we have that $u$ occurs in $t$. Let $st=st_1ut_2$. Then, there are states $p,r\in Q$ in $G'$ such that $\delta'(q_0,st_1) = p$, $\delta'(p,u) = p'$, and $\delta'(p',t_2) = r'$. However, by the construction, this means that $\delta(q_0,st_1)=p$ and $\delta(p,t_2)=r$ in $G$, and hence $P_u(st)=st_1t_2\in L(G)$.

    Since $G$ is $k$-SSO with respect to $Q_S$ and $P$, there exists a string $w\in L(G)$ such that $P(P_u(st))=P(w)$ and, for every prefix $w'$ of $w$, if $|P(w)|-|P(w')|\leq k$, then $\delta(q_0, w')\notin Q_S$. 
    Since $P'(st)=P(P_u(st))=P(w)$, we define $xy=w$ to be a (fixed) decomposition of $w$ such that $P'(s)=P(x)$ and $P'(t)=P(y)$. Then, $|P(w)|-|P(x)| = |P'(st)|-|P'(s)|=|P'(t)| \le k$, which implies that $\delta(q_0,x) = \delta'(q_0,x) = q$ for some state $q$ that is not secret in $G$. Therefore, the transition $(q,u,q')\in \delta'$, and hence $\delta'(q_0,xu)=q' \in Q_{NS}'$. 
    Since the state $\delta(q_0,xy')\notin Q_S$ for every prefix $y'$ of $y$, because $xy'$ is a prefix of $w$ with $|P(w)|-|P(xy')|\le k$, the computation of $\delta(q,y)$ in $G$ does not go through a secret state. Therefore, the same sequence of transitions is enabled in $G'$ from state $q'$.
    Setting now $s'=xu$ and $t'=y$ implies that $P'(s)=P'(s')$, $P'(t)=P'(t')$,  $\delta'(q_0,s')\in Q_{NS}'$, and $\delta'(q_0,s't')$ is defined, which proves that $G'$ is weakly $k$\nobreakdash-step opaque.

    To prove the other direction, we assume that $G$ is not $k$\nobreakdash-SSO with respect to $Q_S$ and $P$, and we show that $G'$ is not $k$-SO with respect to $Q_S'$, $Q_{NS}'$, and $P'$. To this end, we need to show that there is a string $st \in L(G')$ such that $|P'(t)|\leq k$, $\delta'(q_0,s)\in Q_S'$, and for every $s't' \in L(G')$ such that $P'(s)=P'(s')$ and $P'(t)=P'(t')$, the state $\delta'(q_0,s')\notin Q_{NS}'$.
 
    However, since $G$ is not $k$-SSO with respect to $Q_S$ and $P$, there exists a string $v\in L(G)$ such that, for every string $w\in L(G)$ with $P(w)=P(v)$, there is a prefix $w'$ of $w$ such that $|P(w)|-|P(w')|\leq k$ and $\delta(q_0,w')\in Q_S$. In particular, there is a prefix $v'$ of $v$ such that $|P(v)|-|P(v')|\leq k$ and $\delta(q_0,v')\in Q_S$.

    Let $xy=v$ be the decomposition of $v$ such that $y$ is the longest suffix of $v$ containing at most $k$ observable events. We set $s=x$ and $t=y$, for which we have that $|P(t)|\leq k$, $\delta'(q_0,st)$ is defined, and, since neither $s$ nor $t$ contains the event $u$, the state $\delta'(q_0,s)\in Q_S'$. It remains to show that for every string $s't'\in L(G')$ with $P'(s') = P'(s)$ and $P'(t')= P'(t)$, the state $\delta'(q_0,s')\notin Q_{NS}'$. We distinguish two cases. 
  
    In the first case, we assume that $\delta(q_0,P^{-1}P(s))\cap Q_{NS}=\emptyset$, and we consider any string $s't'\in L(G')$ such that $P'(s')=P'(s)$ and $P'(t')=P'(t)$. 
    If $s'$ does not contain the event $u$, then $s'\in P^{-1}P(s)$, and therefore $\delta'(q_0,s') = \delta(q_0,s') \in Q_S\subseteq Q_S'$.
    If, on the other hand, $s'$ contains the event $u$, then $s'=s_1us_2$ where neither $s_1$ nor $s_2$ contains the event $u$. But then $\delta'(q_0,s')=r'$, where $r'$ is a copy of $r=\delta(q_0,s_1s_2)$. Since $s_1s_2\in P^{-1}P(s)$, the state $r=\delta(q_0,s_1s_2)\in Q_S$, and hence $r' \notin Q_{NS}'$ by construction. In both cases, $\delta'(q_0,s')\notin Q_{NS}'$, which was to be shown.

    In the second case, let $\delta(q_0,P^{-1}P(s))\cap Q_{NS} = Z \neq \emptyset$, and consider any string $s't'\in L(G')$ with $P'(s')=P'(s)$ and $P'(t')=P'(t)$. Using the projection $P_u$ removing the event $u$, we set $z:=P_u(s't')\in L(G)$. Recall that the string $st$ does not contain the event $u$, that is, $P'(st)=P(st)$, and therefore $P(z) = P(P_u(s't')) = P'(s't') = P'(st) = P(st) = P(v)$. Since $G$ is not $k$-SSO with respect to $Q_S$ and $P$, there is a prefix $z'$ of $z$ such that $|P(z)|-|P(z')|\leq k$ and $q_s:=\delta(q_0,z')\in Q_S$. In particular, by the choice of $s$, we have that $|P(s)|\le |P(z')|$.
    Furthermore, $G$ is normal, and hence there is no nonsecret state reachable from the secret state $q_s$ by a sequence of unobservable events.
     
    In particular, the prefix $P_u(s')$ of the string $z=P_u(s')P_u(t')$ satisfies $P(P_u(s')) = P'(s') = P'(s)=P(s)$, where the last equality comes from the fact that $s$ does not contain the event $u$. Then $P_u(s') \in P^{-1}P(s)$, and hence if $\delta(q_0,P_u(s'))\in Q_{NS}$, then $\delta(q_0,P_u(s'))\in Q_{NS}\cap Z$.
    Thus, assume that $\delta(q_0,P_u(s'))\in Q_{NS}\cap Z$. Then, the string $P_u(s')$ is a strict prefix of $z'$; 
    otherwise, if $z'$ was a strict prefix of $P_u(s')$, then we would have that $|P(z')|\le |P(P_u(s'))| = |P(s)|$, which, together with $|P(s)|\le |P(z')|$, would give that $|P(z')|=|P(P_u(s'))|=|P(s)|$, and hence the nonsecret state $q_{ns} = \delta(q_0,P_u(s'))$ would be reachable from the secret state $q_s$ by a sequence of unobservable events, which is a contradiction with the normality of $G$.
    Consequently, generating the string $P_u(t')$ from the state $q_{ns}$, $G$ must go through the secret state $q_s$. In other words, $q_s$ is reachable from the state $q_{ns}$ by a prefix of $P_u(t')$.

    Thus, in $G'$, state $\delta'(q_0,s')\in \{q_{ns},q_{ns}'\}$, where $q_{ns}'\in Z'=\{q'\mid q\in Z\}\subseteq Q_{NS}'$. 
    If $\delta'(q_0,s')=q_{ns} \in Q_S'$, we are done. 
    If $\delta'(q_0,s')=q_{ns}' \in Q_{NS}'$, we show that $\delta'(q_{ns}',t')$ is undefined, which contradicts the assumption that $s't' \in L(G')$, and hence $\delta'(q_0,s')=q_{ns}' \in Q_{NS}'$ cannot happen. Indeed, if $\delta'(q_0,s')\in Q_{NS}'$, then $P_u(t')=t'$. Since the computation of $\delta(q_{ns},t')=\delta(q_{ns},P_u(t'))$ in $G$ goes through the secret state $q_s$, the computation of $\delta'(q_{ns}',t')$ in $G'$ has to go through the state $q_s'$, which is the primed copy of the state $q_s$. But the computation $\delta'(q_{ns}',t')$ is performed in the automaton $G_{ns}$, which is obtained from $G$ by removing all secret states and corresponding transitions. Since $q_s'$ is a copy of a secret state, it does not exist in $G_{ns}$, and hence it does not belong to $Q_{NS}'$. Therefore, $\delta'(q_{ns}',t')$ is undefined.
    
    We have thus shown that $G'$ is not $k$\nobreakdash-step opaque.
  \end{proof}

  Both Construction~\ref{consNorm} and Construction~\ref{consSSOtoSO} are polynomial and preserve the number of observable events. Therefore, Theorem~\ref{thm_ksso-kso} and the existing results \citep[Table~1]{BalunMasopust2021} immediately imply the following result.
  
  \begin{thm}
    Deciding whether a given deterministic DES $G$ is strongly $k$-step opaque is (i) \coNP-complete if $G$ has only a single observable event, and (ii) \PSpace-complete if $G$ has at least two observable events.
    \hfill$\qed$
  \end{thm}
  
  The motivation to consider the case of a single observable event comes from the timed discrete-event systems framework of \cite{BrandinW1994}, where the only observable event is the tick of the global clock.

\subsection{Verifying Strong {\it k}-Step Opacity}
  Theorem~\ref{thm_ksso-kso} gives us a clue how to verify strong $k$-step opacity of a given deterministic DES with the help of the verification algorithm for weak $k$-step opacity from Section~\ref{ksoAlgo}. This idea is formulated as Algorithm~\ref{alg2}. 
  Before analyzing the complexity of Algorithm~\ref{alg2} and comparing it with the existing algorithms, we illustrate it by two examples. 
  
  \begin{algorithm}\hrule height .08em\vspace{-5pt}
    \caption{Verification of strong $k$-step opacity}\label{alg2}
    \begin{algorithmic}[1]
      \vspace{2pt}\hrule\vspace{3pt}
      \Require A deterministic DES $G=(Q,\Sigma,\delta,q_0)$, $Q_S\subseteq Q$, $\Sigma_o\subseteq \Sigma$, and $k\in\mathbb{N}_\infty$.

      \Ensure {\tt true} if and only if $G$ is strongly $k$-step opaque with respect to $Q_S$ and $P\colon \Sigma^*\to\Sigma_o^*$

      \State Let $G_{norm}$ be the normalization of $G$ by Construction~\ref{consNorm}
      
      \State Transform $G_{norm}$ to $G'$ by Construction~\ref{consSSOtoSO}
      
      \State Use Algorithm~\ref{alg1} on $G'$ with the set of secret states $Q_S'$, the set of nonsecret states $Q_{NS}'$, observable events $\Sigma_o$, and $k$
      
      \State \Return the answer of Algorithm~\ref{alg1}
    \end{algorithmic}
    \hrule height .08em
  \end{algorithm}

  \begin{figure}[b]
    \centering
    \includegraphics[align=c,scale=0.63]{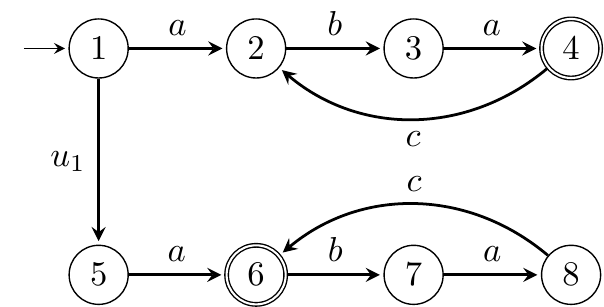}
    $\Longrightarrow$
    \includegraphics[align=c,scale=0.63]{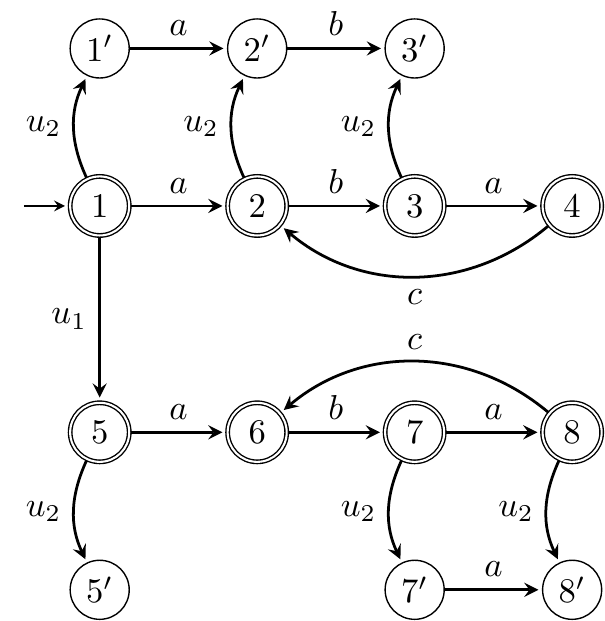}
    \caption{A DES $G$ (left) and $G'$ (right), which is the result of Construction~\ref{consSSOtoSO}.}
    \label{ex:ksso-2}
  \end{figure}

  As the first example, we adopt the DES $G$ from \cite{Falcone2014} depicted in Figure~\ref{ex:ksso-2} (left), where events $a$, $b$, $c$ are observable, event $u_1$ is unobservable, and states $4$ and $6$ are secret. \cite{Falcone2014} claimed that $G$ is strongly one-step opaque. However, it is not the case, as we show by using our transformation to weak $k$\nobreakdash-step opacity.
   
  Since $G$ is normal, Algorithm~\ref{alg2} proceeds directly to the application of Construction~\ref{consSSOtoSO}, which results in the DES $G'$ depicted in Figure~\ref{ex:ksso-2} (right). Namely, $G'$ was constructed from $G$ by adding six new nonsecret states and one new unobservable event $u_2$, and by making states $1$ through $8$ secret, that is, $Q_{S}'=\{1,2,3,4,5,6,7,8\}$ and $Q_{NS}'=\{1',2',3',5',7',8'\}$.
  Applying Algorithm~\ref{alg1} to $G'$, $Q_{S}'$, $Q_{NS}'$, $\Sigma_o=\{a,b,c\}$, and $k=1$ results in the observer $G'^{obs}$ of $G'$ and the automaton $H$ depicted in Figure~\ref{ex:ksso-2-obs}.
  \begin{figure}
    \centering
    \includegraphics[align=c,scale=0.67]{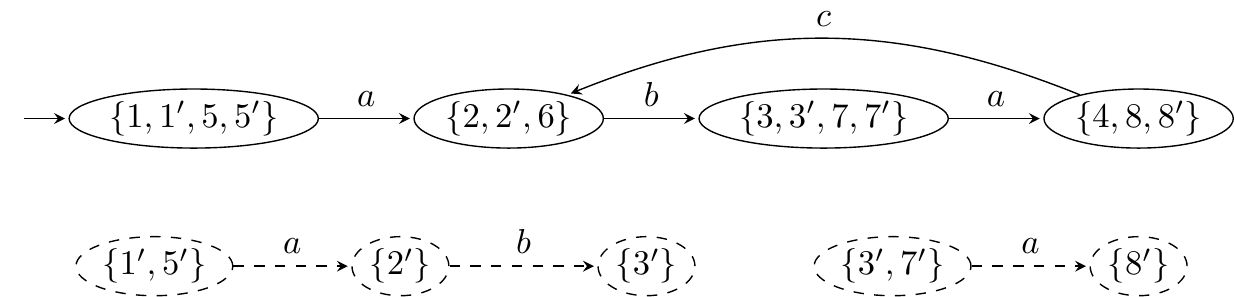}
    \caption{The observer $G'^{obs}$ of $G'$, the solid part. The automaton $H$ forming the relevant part of the full observer of $G'$ is obtained from $G'^{obs}$ by adding the dashed part; state $\emptyset$ and the transitions to it are not depicted.}
    \label{ex:ksso-2-obs}
  \end{figure}
  The set $Y$ and the part of $\C_1 = P(G') \times H$ reachable from the states of $Y$ in one step are depicted in Figure~\ref{ex:ksso-2-c}.
  Since, e.g., state $(2,\emptyset)$ is reachable in $\C_1$ in one step from the state $(4,\{8'\})\in Y$, $G'$ is not weakly one\nobreakdash-step opaque. By Theorem~\ref{thm_ksso-kso}, $G$ is not strongly one-step opaque.
  Indeed, observing the string $abac$ in $G$, the intruder reveals that $G$ is either in the secret state $6$ at that time instant or must have been in the secret state $4$ one step ago.

  \begin{figure}
    \centering
    \includegraphics[align=c,scale=0.62]{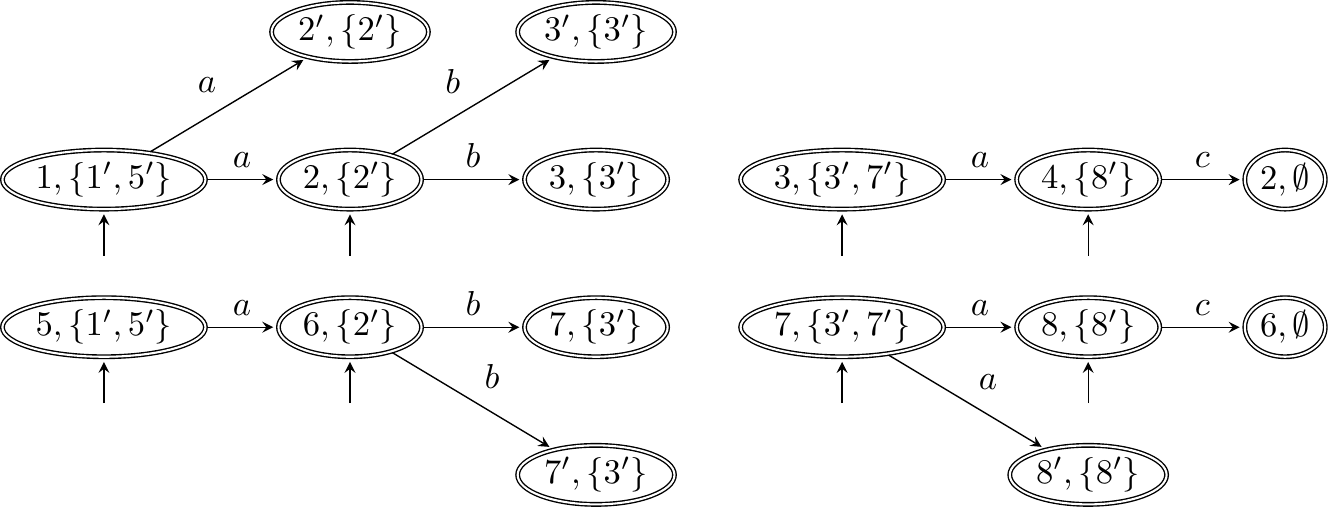}
    \caption{The part of $C_1$ consisting of states reachable from the states of $Y$ in one step. The states of $Y$ are denoted by little arrows.}
    \label{ex:ksso-2-c}
  \end{figure}

  As the second example, we consider the DES $\tilde G$ obtained from $G$ by replacing the transitions $(4,c,2)$ and $(8,c,6)$ by the transitions $(4,c,3)$ and $(8,c,7)$, respectively, see Figure~\ref{ex:ksso-1}. 
  \begin{figure}[b]
      \centering
      \includegraphics[align=c,scale=0.63]{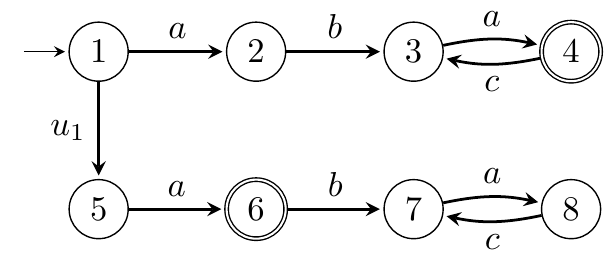}
      $\Longrightarrow$
      \includegraphics[align=c,scale=0.63]{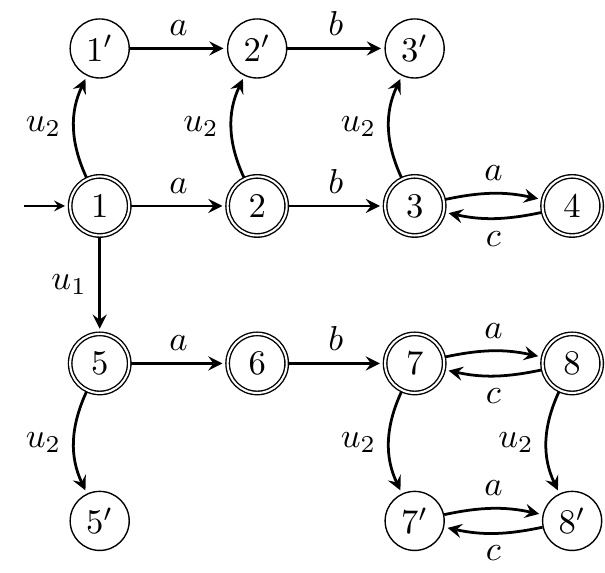}
    \caption{A DES $\tilde G$ (left) and the DES $\tilde G'$ obtained by Construction~\ref{consSSOtoSO} (right).}
    \label{ex:ksso-1}
  \end{figure}
  Then, Construction~\ref{consSSOtoSO} results in the DES $\tilde G'$ with $\tilde Q_S'=\{1,2,3,4,5,6,7,8\}$ and $\tilde Q_{NS}'=\{1',2',3',5',7',8'\}$. Applying Algorithm~\ref{alg1} to $\tilde G'$, $\tilde Q_{S}'$, $\tilde Q_{NS}'$, $\Sigma_o=\{a,b,c\}$, and $k=1$ results in the observer $\tilde{G}'^{obs}$ of $\tilde G'$ and the automaton $\tilde H$ depicted in Figure~\ref{ex:ksso-1-obs}. The set $\tilde Y$ and the part of $\C_2 = P(\tilde{G}') \times \tilde H$ reachable from the states of $Y$ in one step are depicted in Figure~\ref{ex:ksso-1-c}. Since no state of the form $(q,\emptyset)$ is reachable from a state of $\tilde Y$ in one step, $\tilde G'$ is weakly one\nobreakdash-step opaque. By Theorem~\ref{thm_ksso-kso}, $\tilde G$ is strongly one-step opaque. 

  \begin{figure}
    \centering
    \includegraphics[align=c,scale=0.67]{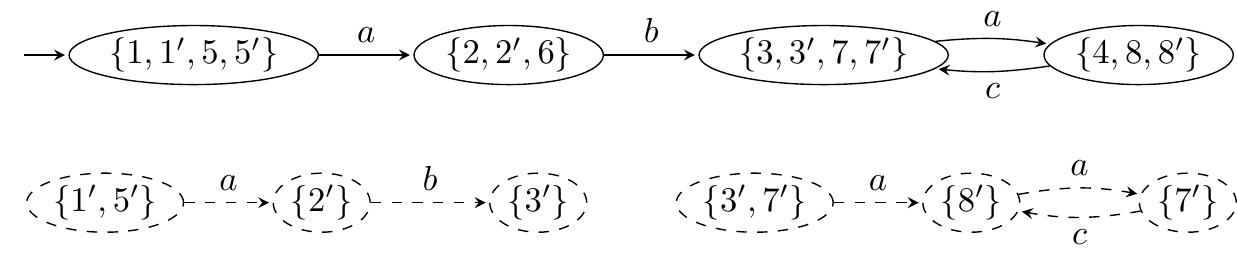}
    \caption{The observer $\tilde G'^{obs}$ of $\tilde G'$, the solid part. The automaton $\tilde H$ forming the relevant part of the full observer of $\tilde G'$ is obtained from $\tilde G'^{obs}$ by adding the dashed part; state $\emptyset$ and the transitions to it are not depicted.}
    \label{ex:ksso-1-obs}
  \end{figure}

  \begin{figure}
    \centering
    \includegraphics[align=c,scale=0.58]{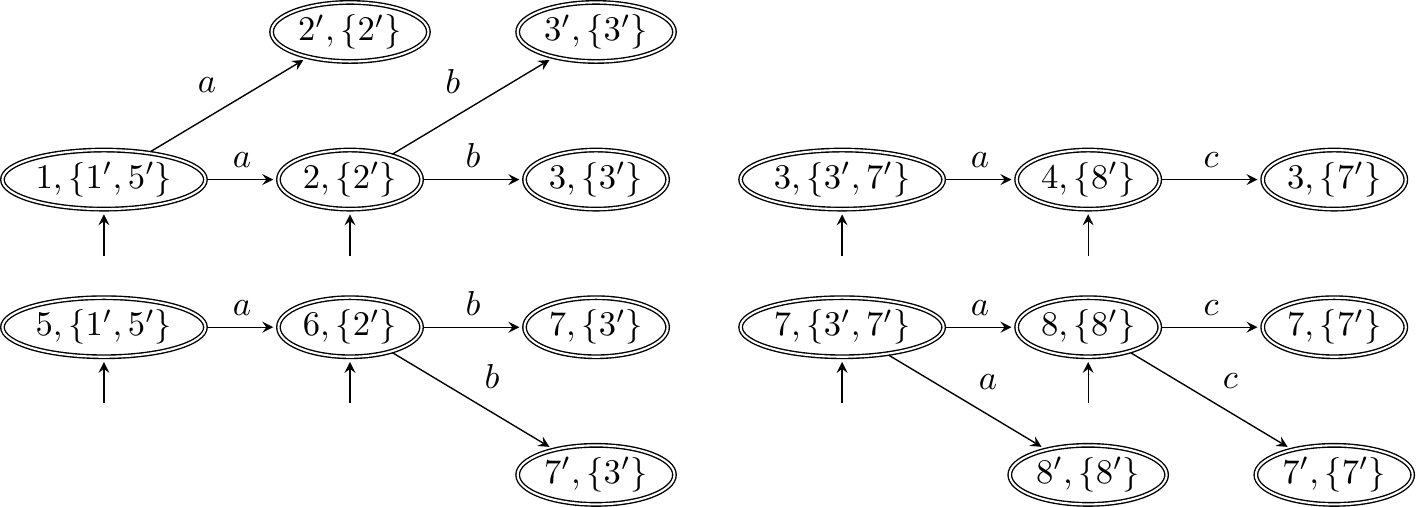}
    \caption{The part of $C_2$ consisting of states reachable from the states of $\tilde Y$ in one step. The states of $\tilde Y$ are denoted by little arrows.}
    \label{ex:ksso-1-c}
  \end{figure}

\subsection{Complexity Analysis of Algorithm~\ref{alg2}}
  Algorithms verifying strong $k$-step opacity have been investigated in the literature. In particular, \cite{Falcone2014} designed an algorithm based on a $k$-delay trajectory estimation, but they did not analyze its complexity. The complexity analyses in the literature are inconsistent. While \cite{ma2021} state that the complexity is $O(\ell 2^{n^2+n})$, where $n$ is the number of states and $\ell$ is the number of observable events of the verified deterministic DES, \cite{Wintenberg2021} state that the state complexity is $O((\ell+1)^k 2^n)$. According to \cite[Definition~7]{Falcone2014}, however, the $k$-delay trajectory estimator has $O(2^{n^{k+1} \cdot 2^{k}})$ states. 

  Recently, \cite{ma2021} designed another algorithm with complexity $O(\ell 2^{(k+2)n})$, and even more recently, \cite{Wintenberg2021} discussed and experimentally compared algorithms based on (i) the secret observer with complexity $O(\ell (k+3)^n)$, on (ii) the reverse comparison with complexity $O((n+m)(k+1)2^n)$, where $m\le \ell n^2$ is the number of transitions in the involved projected NFA, and on (iii) the construction of the $k$\nobreakdash-delay trajectory estimator of \cite{Falcone2014}, which they claim to be of complexity $O(\ell (\ell+1)^k 2^n)$.

  We now analyze the complexity of Algorithm~\ref{alg2} and show that its worst-case complexity is better than the complexity of existing algorithms. 
  Namely, we show that the space and time complexity of Algorithm~\ref{alg2} is $O(n2^n)$ and $O((n + m)2^n)$, respectively, where $n$ is the number of states of $G$ and $m$ is the number of transitions of $P(G)$. Notice that the complexity does not depend on the parameter $k$. 
  
  Before we prove this result, notice that $m \le \ell n^2$, where $\ell$ is the number of observable events. Since $\ell n^2$ is the maximum number of transitions in an $n$-state NFA with $\ell$ events, $m$ is often significantly smaller than $\ell n^2$.
  
  For a deterministic DES with $n$ states, Construction~\ref{consNorm} results in a normalized DES with up to $2n$ states, and hence it may seem that the observer of the normalized DES could have up to $2^{2n}$ states. The following lemma shows that the observer of the normalized DES has in fact at most $2^n$ states.
  
  \begin{lem}\label{lemma1}
    Let $G$ be an $n$-state deterministic DES, and let $G_{norm}$ be its normalization obtained by Construction~\ref{consNorm}. Then, the observer of $G_{norm}$ has at most $2^{n}$ states.
  \end{lem}
  \begin{proof}
    Let $G=(Q,\Sigma, \delta, q_0)$ be a deterministic DES with $n$ states, and let $\Sigma_{uo}$ be the set of unobservable events. The application of Construction~\ref{consNorm} on $G$ results in the deterministic DES $G_{norm}=(Q_n,\Sigma, \delta_{n}, q_0)$, 
    where $Q_n \subseteq Q\cup Q'$ and $Q'=\{q' \mid q\in Q\}$ is a disjoint copy of $Q$. All states of $Q_n$ are reachable in $G_{norm}$ by construction.
    The observer $G_{norm}^{obs}=(X_{obs},\Sigma_o,\delta_{obs},x_0)$ of $G_{norm}$ is defined as follows. The set of states is the subset of the power set of $Q_n$, namely, $X_{obs} \subseteq 2^{Q_n}$. The initial state is the unobservable reach (UR) of the initial state of the automaton $G_{norm}$, that is, $x_0 := UR(q_0) = \delta_n(q_0,\Sigma_{uo}^*)$. The transition function $\delta_{obs}$ is defined for every $X\in X_{obs}$ and every observable event $a\in \Sigma_o$ as the unobservable reach of the states reachable in $G_{norm}$ from the states of $X$ by the event $a$, that is, $$\delta_{obs}(X,a) := UR(\delta_n(X,a))$$ where, for every $Y\subseteq Q_n$, $UR(Y)=\delta_n(Y,\Sigma_{uo}^*)$.
    By item (2) of Lemma~\ref{lemma:delta_norm}, $$\delta_n(X,a)\subseteq Q\,,$$ and hence every state of the observer of $G_{norm}$ is uniquely determined by a subset of $Q$. 
    In particular, we define an injective mapping $f\colon X_{obs}\to 2^Q$ assigning subsets of $Q$ to the states of the observer of $G_{norm}$ as follows: $f(x_0)=\{q_0\}$, and for every state $Y\neq x_0$ of the observer of $G_{norm}$, we pick and fix a state $X\in X_{obs}$ such that $\delta_{obs}(X,a)=Y$, for some observable event $a\in\Sigma_o$, and we define $f(Y)=\delta_n(X,a)$. Such a state $X$ exists because every state of the observer is reachable. Then, $Y=\delta_{obs}(X,a) = UR(\delta_n(X,a)) = UR(f(Y))$, and we have that if $f(Y_1)=f(Y_2)$, then $Y_1 = UR(f(Y_1)) = UR(f(Y_2)) = Y_2$, which shows that the mapping $f$ is injective.
    Consequently, the number of states of the observer of $G_{norm}$ is bounded by the number of subsets of the set $Q$, which is $2^n$.
  \end{proof}

  Notice that Lemma~\ref{lemma1} does not claim that the number of states of the observer of $G$ and of the observer of its normalization $G_{norm}$ coincide. It only provides an upper bound on the worst-case complexity.
  
  Similarly, for a normal deterministic DES $G$ with $n$ states, Construction~\ref{consSSOtoSO} results in a deterministic DES $G'$ with up to $2n$ states. The second lemma shows that the observer of $G'$ has as many states as the observer of $G$.
  
  \begin{lem}\label{lemma2}
    Let $G$ be a normal deterministic DES with $n$ states, and let $G'$ be obtained from $G$ by Construction~\ref{consSSOtoSO}. Then, the numbers of states of the observer of $G'$ and of the observer of $G$ coincide. 
  \end{lem}
  \begin{proof}
    Let $G=(Q,\Sigma,\delta,q_0)$ be a normal deterministic DES, and let $G'=(Q\cup Q_{NS}',\Sigma\cup\{u\},\delta', q_0)$ be the DES obtained from $G$ by Construction~\ref{consSSOtoSO}. Recall that $G'$ is obtained as a disjoint union of $G$ and $G_{ns}$, where $G_{ns}$ is a copy of $G$ without the secret states and the corresponding transitions, $Q_{NS}'=\{q' \mid q\in Q_{NS}\}$ is a copy of $Q_{NS}$ disjoint from $Q$, and the event $u$ is unobservable.
    For every reachable state $S$ of the observer of $G'$, we show that $S$ contains a state $p'\in Q_{NS}'$ if and only if $S$ contains the corresponding state $p\in Q_{NS}$. Consequently, the observer of $G'$ and the observer of $G$ have the same number of states.
    
    To prove one direction, let $S$ be a reachable state of the observer of $G'$. If $S$ contains a state $p\in Q_{NS}$, then the unobservable transition $(p,u,p')$ of $G'$ implies that $S$ also contains the state $p'\in Q_{NS}'$.
    
    To prove the other direction, let $S$ be a reachable state of the observer of $G'$, and assume that a state $p'\in Q_{NS}'$ belongs to $S$. Then, for every string $w\in P'(L(G'))$ under which the state $S$ is reachable from the initial state $\{q_0\}$ in the observer of $G'$, there exists a string $w'\in L(G')$ such that $P'(w')=w$ and $\delta'(q_0,w')=p'$. Since $q_0\in Q$, $p'\in Q_{NS}'$, and every string of $L(G')$ contains at most one occurrence of the event $u$, we can partition the string $w'=w_1uw_2$ so that $\delta'(q_0,w_1)=r$, $\delta'(r,u)=r'$, and $\delta'(r',w_2)=p'$, for some state $r\in Q$. However, $\delta'(r',w_2)=p'$ is executed in $G_{ns}$, which is obtained from $G$ by removing all secret states. Therefore, $\delta(r,w_2)=p$ must be defined in $G$. Altogether, we have shown that $\delta(q_0,w_1w_2)=p$ is defined in $G$, and hence the string $w_1w_2\in L(G)$. Since $w=P'(w')=P'(w_1w_2)$, we have shown that $p\in S$. 
  \end{proof}

  We can now prove the following result analyzing the complexity of Algorithm~\ref{alg2}.
  \begin{thm}
    The space and time complexity of Algorithm~\ref{alg2} is $O(n2^{n})$ and $O((n+m)2^{n})$, respectively, where $n$ is the number of states of $G$ and $m$ is the number of transitions of $P(G)$, that is, $m \le \ell n^2$, where $\ell$ is the number of observable events.
  \end{thm}
  \begin{proof}
    Let $G$ be an $n$-state deterministic DES. In the first step, we construct the normalization $G_{norm}$ of $G$ with at most $2n$ states, the observer of which has at most $2^n$ states by Lemma~\ref{lemma1}. 
    Then, we apply Algorithm~\ref{alg1} to $G'$ obtained from $G_{norm}$ by Construction~\ref{consSSOtoSO}. In particular, by Lemma~\ref{lemma2}, we compute the observer $G'^{obs}$ of $G'$ with at most $2^{n}$ states, and the projected automaton $P(G')$ with at most $4n$ states. Then, for every reachable state $X$ of $G'^{obs}$, and for every $x\in X\cap Q_S'$, we add the pair $(x,X \cap Q_{NS}')$ to the set $Y$. This computation takes time $O(n 2^{n})$. Afterwards, we construct the automaton $H$ as the part of the full observer of $G'$ that is accessible from the states of the second components of $Y$. Since $H$ consists only of the subsets of $Q_{NS}'$, of which there is at most $2^n$, the automaton $H$ has $O(2^n)$ states. The automaton $\C = P(G') \times H$ thus has $O(n2^n)$ states and $O(m 2^n)$ transitions, the sum of which is the time complexity of the BFS applied to mark states of $\C$ reachable from the states of $Y$ in at most $k$ steps. Therefore, the state complexity of Algorithm~\ref{alg2} is $O(n2^{n})$ and the time complexity is $O(n2^{n} + (n+m)2^{n})=O((n+m)2^{n})$.
  \end{proof}

  Comparing the complexity $O((n+m) 2^n)$ of Algorithm~\ref{alg2} with the complexity of the existing algorithms, the reader may see that (1) the complexity of Algorithm~\ref{alg2} does not depend on the parameter $k$, and (2) it is better than the complexity of the existing algorithms, because the minimum of the worst-case complexities $O(\ell 2^{n^{k+1} \cdot 2^{k}})$, $O(\ell 2^{(k+2)n})$, $O(\ell (k+3)^n)$, and $O((n+m)(k+1)2^n)$ of the existing algorithms discussed at the beginning of this subsection is
  $O((n+m)2^n)$ for $k=1$, and 
  $O((n+m)(k+1)2^n) = O((n+m) 2^{2n})$ for $k\in O(2^n)$. 
  Notice that the minimum worst-case complexity for large $k$ is significantly higher than the complexity $O((n+m) 2^n)$ of Algorithm~\ref{alg2}. 
  In fact, the complexity of Algorithm~\ref{alg2}, and the minimum worst-case complexity of the existing algorithms for very small $k$, coincide. However, while the existing algorithms can handle only inputs with a very small $k$ with this complexity, our algorithm can handle inputs with $k$ of arbitrary value with this complexity.
  Consequently, our algorithm improves the complexity of the verification of strong $k$-step opacity.

\section{Conclusions}
  We investigated and discussed the relationship between the notions of weak and strong $k$\nobreakdash-step opacity.
  We designed an algorithm verifying weak $k$\nobreakdash-step opacity that, compared with the existing algorithms, does not depend on the parameter $k$, and has a lower worst-case complexity than the existing algorithms.
  We further discussed strong $k$-step opacity and transformed it to weak $k$\nobreakdash-step opacity in linear time, obtaining thus an algorithm to verify strong $k$\nobreakdash-step opacity. Again, this algorithm does not depend on the parameter $k$, and has a lower worst-case complexity than the existing algorithms.

  Finally, we point out that the complexity of our algorithms may be further optimized. For instance, rather than eliminating $\eps$-transitions in $P(G)$, we may suitably redefine the product of two NFAs for NFAs with $\eps$-transitions. In this case, the parameter $m$ in the complexity formula would be bounded by $\ell n$ rather than $\ell n^2$. We leave these and further optimizations for the future work.

\begin{ack}
  We gratefully acknowledge suggestions and comments of the anonymous referees.
\end{ack}

\appendix
\section{Relation between Strong and Weak $k$-Step Opacity}\label{ma}
  As already mentioned in the introduction, \cite{ma2021} pointed out that strong $k$-step opacity and weak $k$-step opacity are incomparable in the sense that neither strong $k$-step opacity implies weak $k$-step opacity nor the other way round.

  We now show that under the assumption that the set of states is partitioned into secret and nonsecret states, disregarding neutral states, which is a common assumption in the literature~\citep{SabooriH11,Yin2017}, including this paper, the two notions are comparable. This fact was also recently pointed out by \cite{Wintenberg2021}.

  \begin{thm}
    Let $G=(Q,\Sigma,\delta,i)$ be a deterministic DES, and let $k\in\mathbb{N}_\infty$. If $G$ is $k$-SSO with respect to $Q_S$ and $P$, then $G$ is $k$-SO with respect to $Q_S$, $Q-Q_{S}$, and $P$.
  \end{thm}
  \begin{proof}
    Assume that $G$ is $k$-SSO with respect to $Q_S$ and $P$. We show that $G$ is $k$-SO with respect to $Q_S$, $Q_{NS}=Q-Q_S$, and $P$. To this end, we consider any string $st \in L(G)$ with $|P(t)| \leq k$ such that $\delta(i,s)=q_s \in Q_S$; since $st\in L(G)$, we have $\delta(q_s, t)$ is defined. Since $G$ is $k$-SSO, there is a string $w \in L(G)$ such that $P(st)=P(w)$ and for every prefix $w'$ of $w$, if $|P(w)|-|P(w')| \le k$, then $\delta(i,w')\notin Q_S$. Because $P(st)=P(w)$, the string $w$ can be written as $w=s't'$, where $P(s)=P(s')$ and $P(t)=P(t')$. Since $s'$ is a prefix of $w=s't'$, and $|P(s't')|-|P(s')| = |P(t')| = |P(t)| \le k$, we have $\delta(i,s')\notin Q_S$, that is, $\delta(i,s')\in Q_{NS}$. But then the string $s't' \in L(G)$ is such that $P(s)=P(s')$, $P(t)=P(t')$, $\delta(i,s') = q_{ns} \in Q_{NS}$, and $\delta(q_{ns}, t')$ is defined, because $s't'=w\in L(G)$; therefore, $G$ is $k$-SO with respect to $Q_S$, $Q-Q_S$, and $P$.
  \end{proof}

  It can be seen in the previous proof that if there are neutral states, then the implication does not hold in general.

\bibliographystyle{elsarticle-harv}
\bibliography{mybib}

\begin{thebibliography}{23}
\expandafter\ifx\csname natexlab\endcsname\relax\def\natexlab#1{#1}\fi
\providecommand{\url}[1]{\texttt{#1}}
\providecommand{\href}[2]{#2}
\providecommand{\path}[1]{#1}
\providecommand{\DOIprefix}{doi:}
\providecommand{\ArXivprefix}{arXiv:}
\providecommand{\URLprefix}{URL: }
\providecommand{\Pubmedprefix}{pmid:}
\providecommand{\doi}[1]{\href{http://dx.doi.org/#1}{\path{#1}}}
\providecommand{\Pubmed}[1]{\href{pmid:#1}{\path{#1}}}
\providecommand{\bibinfo}[2]{#2}
\ifx\xfnm\relax \def\xfnm[#1]{\unskip,\space#1}\fi
%Type = Article
\bibitem[{Badouel et~al.(2007)Badouel, Bednarczyk, Borzyszkowski, Caillaud and
  Darondeau}]{Badouel2007}
\bibinfo{author}{Badouel, E.}, \bibinfo{author}{Bednarczyk, M.},
  \bibinfo{author}{Borzyszkowski, A.}, \bibinfo{author}{Caillaud, B.},
  \bibinfo{author}{Darondeau, P.}, \bibinfo{year}{2007}.
\newblock \bibinfo{title}{Concurrent secrets}.
\newblock \bibinfo{journal}{Discrete Event Dynamic Systems}
  \bibinfo{volume}{17}, \bibinfo{pages}{425--446}.
%Type = Article
\bibitem[{Balun and Masopust(2021)}]{BalunMasopust2021}
\bibinfo{author}{Balun, J.}, \bibinfo{author}{Masopust, T.},
  \bibinfo{year}{2021}.
\newblock \bibinfo{title}{Comparing the notions of opacity for discrete-event
  systems}.
\newblock \bibinfo{journal}{Discrete Event Dynamic Systems}
  \bibinfo{volume}{31}, \bibinfo{pages}{553--582}.
%Type = Inproceedings
\bibitem[{Balun and Masopust(2022a)}]{BalunMasopust2022smc}
\bibinfo{author}{Balun, J.}, \bibinfo{author}{Masopust, T.},
  \bibinfo{year}{2022}a.
\newblock \bibinfo{title}{On transformations among opacity notions}, in:
  \bibinfo{booktitle}{{IEEE} International Conference on Systems, Man, and
  Cybernetics (SMC)}, pp. \bibinfo{pages}{3012--3017}.
%Type = Article
\bibitem[{Balun and Masopust(2022b)}]{BalunMasopust2022wodes}
\bibinfo{author}{Balun, J.}, \bibinfo{author}{Masopust, T.},
  \bibinfo{year}{2022}b.
\newblock \bibinfo{title}{On verification of weak and strong k-step opacity for
  discrete-event systems}.
\newblock \bibinfo{journal}{{IFAC} {PapersOnLine}} \bibinfo{volume}{55},
  \bibinfo{pages}{108--113}.
%Type = Article
\bibitem[{Brandin and Wonham(1994)}]{BrandinW1994}
\bibinfo{author}{Brandin, B.}, \bibinfo{author}{Wonham, W.},
  \bibinfo{year}{1994}.
\newblock \bibinfo{title}{Supervisory control of timed discrete-event systems}.
\newblock \bibinfo{journal}{IEEE Transactions on Automatic Control}
  \bibinfo{volume}{39}, \bibinfo{pages}{329--342}.
%Type = Article
\bibitem[{Bryans et~al.(2008)Bryans, Koutny, Mazar{\'{e}} and
  Ryan}]{BryansKMR08}
\bibinfo{author}{Bryans, J.W.}, \bibinfo{author}{Koutny, M.},
  \bibinfo{author}{Mazar{\'{e}}, L.}, \bibinfo{author}{Ryan, P.Y.A.},
  \bibinfo{year}{2008}.
\newblock \bibinfo{title}{Opacity generalised to transition systems}.
\newblock \bibinfo{journal}{International Journal of Information Security}
  \bibinfo{volume}{7}, \bibinfo{pages}{421--435}.
%Type = Article
\bibitem[{Bryans et~al.(2005)Bryans, Koutny and Ryan}]{Bryans2005}
\bibinfo{author}{Bryans, J.W.}, \bibinfo{author}{Koutny, M.},
  \bibinfo{author}{Ryan, P.Y.}, \bibinfo{year}{2005}.
\newblock \bibinfo{title}{Modelling opacity using {P}etri nets}.
\newblock \bibinfo{journal}{Electronic Notes in Theoretical Computer Science}
  \bibinfo{volume}{121}, \bibinfo{pages}{101--115}.
%Type = Book
\bibitem[{Cassandras and Lafortune(2021)}]{Lbook}
\bibinfo{editor}{Cassandras, C.G.}, \bibinfo{editor}{Lafortune, S.} (Eds.),
  \bibinfo{year}{2021}.
\newblock \bibinfo{title}{Introduction to Discrete Event Systems}.
\newblock \bibinfo{edition}{3rd} ed., \bibinfo{publisher}{Springer}.
%Type = Book
\bibitem[{Cormen et~al.(2009)Cormen, Leiserson, Rivest and Stein}]{IntroToAlg}
\bibinfo{author}{Cormen, T.H.}, \bibinfo{author}{Leiserson, C.E.},
  \bibinfo{author}{Rivest, R.L.}, \bibinfo{author}{Stein, C.},
  \bibinfo{year}{2009}.
\newblock \bibinfo{title}{Introduction to Algorithms}.
\newblock \bibinfo{publisher}{{MIT} Press}.
%Type = Inproceedings
\bibitem[{Dubreil et~al.(2008)Dubreil, Darondeau and Marchand}]{Dubreil2008}
\bibinfo{author}{Dubreil, J.}, \bibinfo{author}{Darondeau, P.},
  \bibinfo{author}{Marchand, H.}, \bibinfo{year}{2008}.
\newblock \bibinfo{title}{Opacity enforcing control synthesis}, in:
  \bibinfo{booktitle}{Workshop on Discrete Event Systems (WODES)}, pp.
  \bibinfo{pages}{28--35}.
%Type = Article
\bibitem[{Falcone and Marchand(2015)}]{Falcone2014}
\bibinfo{author}{Falcone, Y.}, \bibinfo{author}{Marchand, H.},
  \bibinfo{year}{2015}.
\newblock \bibinfo{title}{Enforcement and validation (at runtime) of various
  notions of opacity}.
\newblock \bibinfo{journal}{Discrete Event Dynamic Systems}
  \bibinfo{volume}{25}, \bibinfo{pages}{531--570}.
%Type = Book
\bibitem[{Hopcroft et~al.(2006)Hopcroft, Motwani and Ullman}]{HopcroftU79}
\bibinfo{author}{Hopcroft, J.E.}, \bibinfo{author}{Motwani, R.},
  \bibinfo{author}{Ullman, J.D.}, \bibinfo{year}{2006}.
\newblock \bibinfo{title}{Introduction to Automata Theory, Languages, and
  Computation}.
\newblock \bibinfo{publisher}{Addison-Wesley}.
%Type = Article
\bibitem[{Jacob et~al.(2016)Jacob, Lesage and Faure}]{JacobLF16}
\bibinfo{author}{Jacob, R.}, \bibinfo{author}{Lesage, J.J.},
  \bibinfo{author}{Faure, J.M.}, \bibinfo{year}{2016}.
\newblock \bibinfo{title}{Overview of discrete event systems opacity: Models,
  validation, and quantification}.
\newblock \bibinfo{journal}{Annual Reviews in Control} \bibinfo{volume}{41},
  \bibinfo{pages}{135--146}.
%Type = Article
\bibitem[{Jir{\'{a}}skov{\'{a}} and Masopust(2012)}]{JiraskovaM12}
\bibinfo{author}{Jir{\'{a}}skov{\'{a}}, G.}, \bibinfo{author}{Masopust, T.},
  \bibinfo{year}{2012}.
\newblock \bibinfo{title}{On a structural property in the state complexity of
  projected regular languages}.
\newblock \bibinfo{journal}{Theoretical Computer Science}
  \bibinfo{volume}{449}, \bibinfo{pages}{93--105}.
%Type = Article
\bibitem[{Lan et~al.(2020)Lan, Tong, Guo and Giua}]{Lan2020}
\bibinfo{author}{Lan, H.}, \bibinfo{author}{Tong, Y.}, \bibinfo{author}{Guo,
  J.}, \bibinfo{author}{Giua, A.}, \bibinfo{year}{2020}.
\newblock \bibinfo{title}{Comments on {\textquotedblleft}{A} new approach for
  the verification of infinite-step and {K}-step opacity using two-way
  observers{\textquotedblright} [{A}utomatica 80 (2017) 162{\textendash}171]}.
\newblock \bibinfo{journal}{Automatica} \bibinfo{volume}{122},
  \bibinfo{pages}{109290}.
%Type = Article
\bibitem[{Lin(2011)}]{Lin2011}
\bibinfo{author}{Lin, F.}, \bibinfo{year}{2011}.
\newblock \bibinfo{title}{Opacity of discrete event systems and its
  applications}.
\newblock \bibinfo{journal}{Automatica} \bibinfo{volume}{47},
  \bibinfo{pages}{496--503}.
%Type = Article
\bibitem[{Ma et~al.(2021)Ma, Yin and Li}]{ma2021}
\bibinfo{author}{Ma, Z.}, \bibinfo{author}{Yin, X.}, \bibinfo{author}{Li, Z.},
  \bibinfo{year}{2021}.
\newblock \bibinfo{title}{Verification and enforcement of strong infinite- and
  k-step opacity using state recognizers}.
\newblock \bibinfo{journal}{Automatica} \bibinfo{volume}{133},
  \bibinfo{pages}{109838}.
%Type = Inproceedings
\bibitem[{Saboori and Hadjicostis(2007)}]{SabooriHadjicostis2007}
\bibinfo{author}{Saboori, A.}, \bibinfo{author}{Hadjicostis, C.N.},
  \bibinfo{year}{2007}.
\newblock \bibinfo{title}{Notions of security and opacity in discrete event
  systems}, in: \bibinfo{booktitle}{{IEEE} Conference on Decision and Control},
  pp. \bibinfo{pages}{5056--5061}.
%Type = Article
\bibitem[{Saboori and Hadjicostis(2011)}]{SabooriH11}
\bibinfo{author}{Saboori, A.}, \bibinfo{author}{Hadjicostis, C.N.},
  \bibinfo{year}{2011}.
\newblock \bibinfo{title}{Verification of {$K$}-step opacity and analysis of
  its complexity}.
\newblock \bibinfo{journal}{{IEEE} Transactions on Automation Science and
  Engineering} \bibinfo{volume}{8}, \bibinfo{pages}{549--559}.
%Type = Article
\bibitem[{Saboori and Hadjicostis(2012)}]{SabooriH12a}
\bibinfo{author}{Saboori, A.}, \bibinfo{author}{Hadjicostis, C.N.},
  \bibinfo{year}{2012}.
\newblock \bibinfo{title}{Verification of infinite-step opacity and complexity
  considerations}.
\newblock \bibinfo{journal}{{IEEE} Transactions on Automatic Control}
  \bibinfo{volume}{57}, \bibinfo{pages}{1265--1269}.
%Type = Article
\bibitem[{Wintenberg et~al.(2022)Wintenberg, Blischke, Lafortune and
  Ozay}]{Wintenberg2021}
\bibinfo{author}{Wintenberg, A.}, \bibinfo{author}{Blischke, M.},
  \bibinfo{author}{Lafortune, S.}, \bibinfo{author}{Ozay, N.},
  \bibinfo{year}{2022}.
\newblock \bibinfo{title}{A general language-based framework for specifying and
  verifying notions of opacity}.
\newblock \bibinfo{journal}{Discrete Event Dynamic Systems}
  \bibinfo{volume}{32}, \bibinfo{pages}{253--289}.
%Type = Inproceedings
\bibitem[{Wong(1998)}]{wong98}
\bibinfo{author}{Wong, K.}, \bibinfo{year}{1998}.
\newblock \bibinfo{title}{On the complexity of projections of discrete-event
  systems}, in: \bibinfo{booktitle}{Workshop on Discrete Event Systems
  (WODES)}, pp. \bibinfo{pages}{201--206}.
%Type = Article
\bibitem[{Yin and Lafortune(2017)}]{Yin2017}
\bibinfo{author}{Yin, X.}, \bibinfo{author}{Lafortune, S.},
  \bibinfo{year}{2017}.
\newblock \bibinfo{title}{A new approach for the verification of infinite-step
  and {K}-step opacity using two-way observers}.
\newblock \bibinfo{journal}{Automatica} \bibinfo{volume}{80},
  \bibinfo{pages}{162--171}.

\end{thebibliography}

\end{document}